\documentclass[runningheads]{llncs}

\usepackage[T1]{fontenc}
\usepackage{graphicx}
\usepackage{amsmath, amssymb}
\usepackage{physics}
\usepackage{booktabs}
\usepackage{array}
\usepackage{algorithm}
\usepackage{algpseudocode}
\usepackage{tikz}
\usetikzlibrary{arrows.meta, positioning, shapes.geometric, backgrounds, fit, calc}
\usepackage{xcolor}
\usepackage{color}
\usepackage{bbding}
\usepackage[hidelinks]{hyperref}

\newcommand{\ab}{\discretionary{}{}{}}
\begin{document}

\title{QML-PipeGuard: Detecting Channel-Substitution Attacks on
Quantum Machine Learning Pipelines in Untrusted Clouds}

\titlerunning{Detecting Channel-Substitution Attacks on QML Pipelines}

\author{Esra Yeniaras\orcidID{0000-0002-2973-7597}\textsuperscript{(\Envelope)}}
\authorrunning{E. Yeniaras}
\institute{Independent Researcher, T\"urkiye\\
\email{esramath@gmail.com}}

\maketitle

\begin{abstract}
As machine-learning workloads move onto quantum hardware they inherit a trust problem classical AI deployments have largely solved and quantum ones have not: quantum machine learning (QML) is already sold as a cloud service, and a customer cannot check that the circuit they paid for is actually the circuit that ran. We show the gap is exploitable. A malicious or cost-cutting provider can substitute the declared quantum channel with one that reproduces every accuracy test the customer performs, passes the single-observable checks used in practice, and still differs from the declared computation at the channel level. We call this a \emph{channel-substitution attack} and instantiate it in three deployments: audit evasion, service downgrade, and result forgery in reproducibility review. QML-PipeGuard detects the attack at runtime from measurement statistics the customer already collects. Its guarantee is two-sided, stated in the distance a local verifier measures: any substitution shifting a single-qubit output marginal by $\delta$ in trace norm forces a contract deviation of at least $\delta/\sqrt{3}$, the constant tight and independent of qubit count; the guarantee is relative to the chosen observable family and reference-input set. For local-unitary substitutions the adversary's \emph{blind set} is characterised exactly, as a commutation condition, which gives the defender a rule for choosing challenge inputs that shrinks it. A finite-shot budget rule bounds false accepts and false alarms at once, one tolerance separates attack from the benign drift observed here, and every decision enters a hash-anchored audit trail. On an IBM Heron r2 processor a substituted quantum-kernel intrusion-detection pipeline evades the weak contract and our verifier halts it, while the honest channel stays within tolerance.

\end{abstract}

\section{Introduction}
\label{sec:intro}

Delegating a computation to hardware you do not control is a classical security problem, and classical deployments answer it with mechanisms that bind results to the computation that produced them. As AI workloads move onto quantum hardware that answer does not carry over: for cloud-delivered quantum machine learning (QML), no such binding exists today. A client, human or autonomous agent, uploads a circuit specification; the provider returns counts; nothing ties them to the specification. The provider could have run a cheaper circuit, a differently trained ansatz, or a deliberately altered channel. The client's usual checks, accuracy on a held-out set or agreement on one observable, would not notice. We call this a \emph{channel-substitution attack}.

Its preconditions hold today. QML is sold as a service on IBM Quantum, IonQ, and Quantinuum hardware. The QML-as-a-service threat surface already has its own security literature \cite{kundu2024sokqmlaas,catak2025tqml}, and workloads range from quantum-kernel intrusion detection \cite{havlivcek2019supervised,kalinin2023qiids} to variational chemistry \cite{cao2018potential}. Three deployments make the incentive concrete. \textbf{Audit evasion:} a model owner under AI compliance audit substitutes an audit-passing ansatz that masks bias, the QML analogue of a Dieselgate defeat device \cite{Dieselgate}. \textbf{Service downgrade:} a provider silently substitutes a lower-fidelity channel that still passes customer test queries, the quantum form of lazy-server attacks \cite{golle2001uncheatable}. \textbf{Result forgery:} an investigator reports results from a substitute ansatz that passes reviewer replication checks \cite{hutson2018ai}. In each case the substitute preserves the statistics the victim checks (Appendix~\ref{app:relevance} reviews the regulatory and commercial conditions).

A second, non-adversarial force pushes the same way: noisy hardware drifts between recalibrations, and a defence flagging every deviation is unusable, so a deployable mechanism must separate attack from drift under one rule.

\paragraph{Contributions.}
We present QML-PipeGuard, a runtime defence that detects channel substitution from measurement statistics the customer already collects, and absorbs benign drift within a calibrated tolerance. The verifier characterises the pipeline by its \emph{behavioural fingerprint}, the vector of Pauli expectations on reference inputs it prepares itself, and halts when it departs from the declared one beyond tolerance. The one ingredient imported is the classical design-by-contract discipline \cite{liskov1994behavioral,meyer1992applying}; channel contracts for isolated single-qubit maps appear in \cite{yeniaras2026qcivet}. Everything below is established here for the multi-qubit pipeline setting; that treatment's single-qubit counterparts aside, we found no analogue there or elsewhere in the QML verification literature surveyed:

\begin{itemize}
\item A threat model for channel substitution on the entangling multi-qubit pipeline of a deployed QML model, with adversary capabilities and three deployment scenarios (Section~\ref{sec:threat}).
\item A \emph{two-sided} detection guarantee stated in the distance a local verifier measures (Theorem~\ref{thm:detection}), with a tight constant independent of qubit count. Channel-level (diamond-norm) claims are not certifiable by a local, fixed-input verifier, and we do not make them.
\item An \emph{exact characterisation of the blind set for local-unitary substitutions} (Proposition~\ref{prop:local-ic}): such a substitution evades the contract exactly when every factor commutes with every output marginal on the reference inputs, turning an open-ended soundness gap into a rule for choosing challenge inputs.
\item A finite-shot budget rule bounding false accepts and false alarms simultaneously (Theorem~\ref{thm:budget}), and a drift corollary (Corollary~\ref{cor:drift}) letting one tolerance separate attack from calibration drift.
\item Hardware evaluation on an IBM Heron r2 processor (Section~\ref{sec:experiment}): the substituted pipeline evades a weak single-observable contract and is halted by ours with a wide margin, while the honest channel stays within tolerance.
\end{itemize}

\paragraph{Artefacts.}
The implementation, archived hardware run artefacts (raw counts, fingerprints, audit logs, job IDs), and reproduction instructions are available.\footnote{\label{fn:repo}\url{https://github.com/schrodinket/QML-PipeGuard-v2}}

\section{Related Work}
\label{sec:related}

Five threads border this work; Table~\ref{tab:rw-summary} places each, and none we know of treats channel-level execution integrity as a runtime-verifiable contract.

Three sit at layers our contract does not touch. \emph{Pulse-level calibration} (QuPAD \cite{hu2023qupad}, iterative calibration \cite{quantum2024pulse}, adaptive mitigation \cite{henao2023adaptive}) reshapes gate waveforms below the channel, producing the substrate we then verify. \emph{Drift monitoring} (Q-SafeML \cite{dunn2025qsafeml}) tracks divergence between production and training \emph{input distributions}, an axis orthogonal to channel drift: either can occur without the other. \emph{Robustness verification} (VeriQR \cite{lin2024veriqr}, adversarial QML \cite{lu2020quantum,liao2021robust,west2023benchmarking}) certifies output stability when an adversary perturbs \emph{inputs} to a fixed channel; ours leaves inputs alone and replaces the channel.

\emph{Device fingerprinting} performs substrate authentication: idle tomography \cite{mi2021quantum}, error-syndrome classification \cite{mutolo2025fingerprint}, total-variation-distance noise fingerprints \cite{macneil2025authentication}, and noise-gap identification \cite{wu2024qid} establish \emph{which backend} served a job. We perform channel authentication: given that substrate, whether the executed operations match the declared specification. An adversary who honestly routes the job to the advertised device and then substitutes the circuit on it defeats the former and is caught by the latter. We use ``authentication'' behaviourally throughout: it rests on statistical indistinguishability of measured observables rather than a cryptographic key (Section~\ref{sec:limitations}).

\emph{Contract-based verification} begins with Liskov--Wing subtyping \cite{liskov1994behavioral,meyer1992applying}, extended to circuits \cite{yamaguchi2024dbc}, modules \cite{jin2023scaffml}, and refinement orders \cite{feng2025refinement}. The roadmaps and surveys above catalogue the risk surface without a runtime mechanism; zero-knowledge approaches \cite{laufer2025expproof,scaramuzza2025zkmlops,lee2026zkqml} address \emph{privacy} of inference, not \emph{integrity} of execution. Contract-based treatments of quantum channels known to us target isolated single-qubit maps under exact expectations, with the guarantee stated at the channel level, so its constant bounds a different quantity from the marginal-level $\sqrt{3}$ here \cite{yeniaras2026qcivet}. A deployed QML pipeline changes the mathematics: feature maps entangle, and on $n \geq 2$ qubits a local Pauli family misses correlations, so an evasion set exists and must be characterised, where on one qubit a suitable reference set makes it trivial. Sections~\ref{sec:threat} and~\ref{sec:framework} address that setting throughout: an adaptive white-box threat model, a marginal-level two-sided bound with a tight $n$-independent constant, the exact blind-set characterisation, and a finite-shot layer with two-sided error control.

\begin{table}[tb]
\caption{Position relative to the five adjacent threads. Scope differences are complementary layers of one deployment stack.}
\label{tab:rw-summary}
\centering
\renewcommand{\arraystretch}{1.0}
\scriptsize
\setlength{\tabcolsep}{4pt}
\begin{tabular}{>{\raggedright\arraybackslash}p{2.0cm} >{\raggedright\arraybackslash}p{3.1cm} >{\raggedright\arraybackslash}p{2.9cm} >{\raggedright\arraybackslash}p{2.9cm}}
\toprule
\textbf{Thread} & \textbf{Representative work} & \textbf{Question addressed} & \textbf{Relation to this work} \\
\midrule
Pulse-level calibration & QuPAD \cite{hu2023qupad}, \cite{quantum2024pulse,henao2023adaptive} & Raise gate fidelity below the channel layer & Substrate layer; verified channel sits on top \\
Data-drift monitoring & Q-SafeML \cite{dunn2025qsafeml} & Input distribution vs.\ training distribution & Orthogonal drift axis (input vs.\ channel) \\
Input-robustness verification & VeriQR \cite{lin2024veriqr}, \cite{lu2020quantum,liao2021robust} & Output stability under input perturbation & Different attack surface (inputs vs.\ channel identity) \\
Device fingerprinting & Q-ID \cite{wu2024qid}, \cite{mi2021quantum,mutolo2025fingerprint,macneil2025authentication} & Which hardware executed the job & Substrate authentication; result-level contract here \\
Quantum software contracts & DbC \cite{yamaguchi2024dbc}, ScaffML \cite{jin2023scaffml}, refinement \cite{feng2025refinement}, channel contracts \cite{yeniaras2026qcivet} & Per-circuit/module assertions; isolated-channel contracts & Pipeline-level, finite-shot, dual-mode treatment (Section~\ref{sec:framework}) \\
\bottomrule
\end{tabular}
\end{table}

\section{Preliminaries}
\label{sec:background}

Each pipeline stage is a completely positive trace-preserving (CPTP) channel on a finite-dimensional Hilbert space. A declared channel $\mathcal{E}_A$ carries a \emph{stage specification} $\sigma_A = (H_{\text{spec}}, \mathcal{O}_A, \varepsilon_A, \tau_A)$. The four entries are a hash anchor binding it to an audit trail, the declared observable family, the calibrated tolerance, and a timestamp. A candidate $\mathcal{E}_B$ satisfies the contract on a reference input set $\mathcal{S}$ if
\begin{equation}
\big| \operatorname{Tr}(O \, \mathcal{E}_B(\rho)) - \operatorname{Tr}(O \, \mathcal{E}_A(\rho)) \big| \leq \varepsilon_A
\quad \forall O \in \mathcal{O}_A, \; \forall \rho \in \mathcal{S}.
\label{eq:bg-contract}
\end{equation}
Everything certified below is relative to the pair $(\mathcal{O}_A, \mathcal{S})$; Section~\ref{sec:fw-detection} makes this exact. Throughout, $\sigma_A$ denotes the stage specification that fixes $\mathcal{O}_A$ and $\varepsilon_A$ in \eqref{eq:bg-contract}; the reference set $\mathcal{S}$ is declared alongside it (Section~\ref{sec:threat-trust}).

We target variational quantum classifiers (VQC), quantum kernel methods (QSVM), and quantum neural networks (QNN) \cite{biamonte2017qml,cerezo2021vqa,havlivcek2019supervised,schuld2020kernel,abbas2021qnn,schuld2020nn}. All produce output as the expectation of an observable on a prepared state, which makes \eqref{eq:bg-contract} applicable; QCNNs \cite{cong2019qcnn} and QBMs \cite{amin2018qbm} share this property. The declared pipeline is $\mathcal{P}_A = M \circ W(\theta) \circ U_\phi(x)$. For the QSVM kernel of Section~\ref{sec:experiment} this becomes $\mathcal{P}_A = M \circ U_\phi(x_j)^\dagger \circ U_\phi(x_i)$. Appendix~\ref{app:qml} gives the per-class substitution forms.

Any single-qubit state has a Bloch decomposition $\rho = \tfrac{1}{2}(I + r_x X + r_y Y + r_z Z)$, so $\{X,Y,Z\}$ is informationally complete on $\mathbb{C}^2$ and no strict subset is. The \emph{local} family $\mathcal{O}_A^{\text{local}} = \{X_i, Y_i, Z_i\}_{i=1}^{n}$ used here is not informationally complete on $n$ qubits, missing correlations such as $X_1X_2$, but determines the single-qubit output \emph{marginals} exactly, which is the structural fact the detection theory rests on. The gap is intrinsically multi-qubit: on one qubit the family determines the whole output state, so the blind set of Proposition~\ref{prop:local-ic} holds only unitaries commuting with the reference outputs and collapses to global phases once two of them share no eigenbasis; QML feature maps entangle by construction. Estimating $\langle O \rangle$ to error $t$ at confidence $1-\eta$ costs $O(\|O\|^2\log(1/\eta)/t^2)$ shots \cite{huang2020predicting,nielsen2010quantum}.

\section{Threat Model}
\label{sec:threat}

\subsection{System Model and Trust Boundaries}
\label{sec:threat-trust}

A deployment has three parties. The \textbf{model owner} declares a channel $\mathcal{E}_A$ with observable family $\mathcal{O}_A$, reference set $\mathcal{S}$, and tolerance $\varepsilon_A$. The \textbf{execution environment} runs the channel on hardware and returns outcomes. The \textbf{verifier} requests shots, computes expectations, and logs each round as within-tolerance or halts on a deviation beyond $\varepsilon_A$. Figure~\ref{fig:threat-model} shows the resulting trust boundaries.

The verifier is trusted: it holds the declared specification including the hash anchor $H_{\text{spec}}$, with shot-collection access outside the adversary's control. The execution environment is the locus of the threat: adversary-controlled under attack, honest but drifting otherwise. The model owner is trusted or not depending on scenario. The verifier is a regulator in audit settings, the customer or its agent in cloud settings, or a replication community in academic work (Table~\ref{tab:threat-scenarios}).

\begin{figure}[tb]
\centering
\resizebox{0.92\textwidth}{!}{%
\begin{tikzpicture}[
  font=\small,
  >=Stealth,
  thick,
  party/.style={
    rectangle, rounded corners=3pt, draw, thick,
    minimum width=3.0cm, minimum height=1.0cm,
    align=center
  },
  trusted/.style={
    party, fill=green!10, draw=green!50!black
  },
  optional/.style={
    party, fill=gray!10, draw=gray!60!black
  },
  spec/.style={
    rectangle, draw, dashed, fill=blue!5,
    minimum width=3.2cm, minimum height=1.0cm,
    align=center, font=\footnotesize
  },
  decision/.style={
    rectangle, rounded corners=3pt, draw,
    fill=yellow!15, minimum width=3.2cm, minimum height=0.9cm,
    align=center, font=\footnotesize
  },
  adversary/.style={
    ellipse, draw=red!70!black, fill=red!20, thick,
    minimum width=1.6cm, minimum height=0.7cm,
    align=center, font=\footnotesize\bfseries
  },
  pipestage/.style={
    rectangle, draw, fill=orange!15, thick,
    minimum width=1.5cm, minimum height=0.9cm,
    align=center, font=\footnotesize
  },
  flowlabel/.style={
    font=\footnotesize\itshape, midway, fill=white, inner sep=2pt
  },
  envbox/.style={
    rectangle, draw=red!60!black, thick, dashed,
    inner sep=0.5cm, fill=red!5,
    rounded corners=5pt
  }
]

\node[optional] (owner) {Model Owner\\\footnotesize\itshape(potentially untrusted)};
\node[spec, right=2cm of owner] (spec) {Declared Spec $\sigma_A$\\$(H_{\text{spec}}, \mathcal{O}_A, \varepsilon_A, \tau_A)$};

\node[pipestage, below=3.2cm of owner.south west, anchor=north west] (encode) {$U_\phi(x)$\\encoding};
\node[pipestage, right=0.4cm of encode] (ansatz) {$W(\theta)$\\ansatz};
\node[pipestage, right=0.4cm of ansatz] (measure) {$M$\\measurement};

\begin{pgfonlayer}{background}
\node[envbox, fit=(encode) (ansatz) (measure), label={[font=\footnotesize\bfseries, text=red!70!black, anchor=north west]north west:Execution Environment $\mathcal{E}_A$ or $\mathcal{E}_B \neq \mathcal{E}_A$?}] (envbox) {};
\end{pgfonlayer}

\node[adversary, left=2.0cm of envbox.west] (adv) {$\mathcal{A}$};
\node[trusted, right=1.2cm of envbox.east] (verifier) {Verifier\\\footnotesize\itshape(trusted)};
\node[decision, below=1.0cm of verifier] (decision) {Accept / Halt /\\Log within tolerance};

\draw[->] (owner) -- node[flowlabel] {declares} (spec);
\draw[->] (spec.south) -- ++(0, -0.6) -| (ansatz.north) node[flowlabel, near start, right] {deploy};
\draw[->, red!70!black, thick, dashed] (adv.east) -- (envbox.west)
  node[flowlabel, midway, above, text=red!70!black, yshift=2pt] {may substitute};
\draw[->] (verifier.north west) to[bend left=35]
  node[flowlabel, above, yshift=16pt] {observable queries} (envbox.north east);
\draw[->] (envbox.south east) to[bend left=35]
  node[flowlabel, below, yshift=-4pt] {measurement outcomes} (verifier.south west);
\draw[->, dotted] (spec.east) to[bend left=55]
  node[flowlabel, above, yshift=18pt] {reference} (verifier.north);
\draw[->] (verifier) -- node[flowlabel, right] {contract check} (decision);

\end{tikzpicture}}
\caption{System model and trust boundaries. The model owner declares $\sigma_A$; the pipeline runs inside the execution environment, which an adversary $\mathcal{A}$ may control and substitute $\mathcal{E}_B$ for $\mathcal{E}_A$; the trusted verifier queries observables and decides halt or log-within-tolerance. Green: trusted. Red dashed border: adversary-controllable. Gray: trust depends on scenario.}
\label{fig:threat-model}
\end{figure}

\subsection{Adversary Model}
\label{sec:threat-adv}

$\mathcal{A}$ controls the execution environment and replaces $\mathcal{E}_A$ with a substitute $\mathcal{E}_B$ satisfying three conditions:

\begin{enumerate}
\item[(S1)] \emph{Decision agreement}: $\arg\max f_B(x) = \arg\max f_A(x)$, with $f_A, f_B$ the decision scores of the declared and substituted pipelines, for every $x$ in the verifier's accuracy test set, so the substitute passes accuracy-based audits.
\item[(S2)] \emph{Weak-contract agreement}: $\Delta_{\mathcal{O}_{\text{weak}}, \rho}(\mathcal{E}_A, \mathcal{E}_B) \leq \varepsilon_A$ on a weaker family $\mathcal{O}_{\text{weak}}$, e.g.\ $\{Z\}$-only measurements, for all $\rho \in \mathcal{S}$.
\item[(S3)] \emph{Nontrivial separation}: $\| \mathcal{E}_A - \mathcal{E}_B \|_\diamond \geq \delta_\diamond > 0$ as CPTP maps.
\end{enumerate}

\noindent
(S1) and (S2) capture what a weak verifier cannot see. (S3) states what the substitute gains as a channel, in the diamond norm because that is its strongest metric. What a \emph{local} verifier detects is governed instead by the never-larger marginal separation of Definition~\ref{def:marginal-sep}, and the gap between the two is exactly the blind set (Proposition~\ref{prop:local-ic}). The adversary's operational benefit is scenario-dependent (Table~\ref{tab:threat-scenarios}), not part of the detection analysis.

$\mathcal{A}$ may execute arbitrary CPTP channels, unitary, mixed-unitary, or noise-introducing, and may adapt across rounds, choosing $\mathcal{E}_B^{(t)}$ from the full history. It knows the specification, the observable family, the tolerance, the accuracy test set, the reference set $\mathcal{S}$, and the protocol: a white-box adversary on the specification side. Since $\mathcal{S}$ is known, the blind set is known to $\mathcal{A}$ as well; Proposition~\ref{prop:local-ic} characterises that boundary exactly and Section~\ref{sec:limitations} gives the resulting configuration guidance. It does not control the verifier's classical post-processing or its preparation of reference inputs, and succeeds if the verifier accepts while $\mathcal{E}_B \neq \mathcal{E}_A$.

\paragraph{Benign drift.}
Without adversarial intent the channel still diverges, $\mathcal{E}_B = \mathcal{D} \circ \mathcal{E}_A$ for a perturbation $\mathcal{D}$ from gate-fidelity changes, coherence-time fluctuation, or readout drift, unbounded a priori. Drift is unintentional, statistical, and untargeted; substitution is intentional, structured, and may be optimised against the family. Both are handled by the same contract \eqref{eq:bg-contract}, with $\varepsilon_A$ as the boundary: deviations within tolerance are logged and execution continues, deviations beyond it halt the pipeline. An unusually large drift is not distinguished from an attack, since both produce an out-of-tolerance fingerprint and warrant a halt; root-cause attribution is left to downstream forensics. \emph{Out of scope.} Five surfaces lie outside the model. Input-perturbation adversarial examples are orthogonal \cite{lin2024veriqr}. Training-time poisoning needs training-pipeline defences. Denial of service is a service-level concern. Verifier compromise defeats any contract-based scheme by definition. Classical side channels are not covered. Separately, offline settings with no network-level verifier need delegation machinery \cite{mahadev2018classical,fitzsimons2017unconditionally}; single-organisation deployments with shared trust reduce to the drift-only case.

\begin{table}[tb]
\caption{Deployment scenarios mapped to threat model components. The first two are the QML analogues of the Dieselgate defeat device and of lazy-server attacks on outsourced computation.}
\label{tab:threat-scenarios}
\centering
\renewcommand{\arraystretch}{1.0}
\scriptsize
\setlength{\tabcolsep}{4pt}
\begin{tabular}{>{\raggedright\arraybackslash}p{1.5cm} >{\raggedright\arraybackslash}p{3.2cm} >{\raggedright\arraybackslash}p{3.2cm} >{\raggedright\arraybackslash}p{3.2cm}}
\toprule
 & \textbf{Audit evasion} & \textbf{Service downgrade} & \textbf{Result forgery} \\
\midrule
Adversary & Model owner & Cloud provider & Owner or environment \\
Goal & Pass compliance audit, deploy biased ansatz in production & Cut cost while passing customer test queries & Pass replication check, report other results \\
Substitute & Debiased for audit inputs only & Lower fidelity, accuracy-matched on test set & Matches reviewer cases, not full experiment \\
Verifier & Regulator or internal compliance & Customer-side primitive & Reviewer or replication community \\
\bottomrule
\end{tabular}
\end{table}

\section{The QML-PipeGuard Defence}
\label{sec:framework}

\subsection{Behavioural Fingerprints}
\label{sec:fw-fp}

\begin{definition}[Behavioural fingerprint]
\label{def:fingerprint}
For a CPTP channel $\mathcal{E}$, observables $\mathcal{O} = \{O_1, \dots, O_k\}$ on the output, and reference input $\rho_{\text{ref}}$, the fingerprint is $\mathrm{Fp}_{\mathcal{O}, \rho_{\text{ref}}}(\mathcal{E}) = \big(\operatorname{Tr}(O_1 \mathcal{E}(\rho_{\text{ref}})), \dots, \operatorname{Tr}(O_k \mathcal{E}(\rho_{\text{ref}}))\big)$.
\end{definition}

\noindent
The \emph{contract deviation} is $\Delta_{\mathcal{O}, \rho}(\mathcal{E}_A, \mathcal{E}_B) = \max_{O \in \mathcal{O}} | \operatorname{Tr}(O \mathcal{E}_B(\rho)) - \operatorname{Tr}(O \mathcal{E}_A(\rho)) |$, so $\mathcal{E}_B$ satisfies \eqref{eq:bg-contract} iff $\Delta_{\mathcal{O}, \rho} \leq \varepsilon_A$ on $\mathcal{S}$. A substitution meeting (S1)--(S3) is \emph{sneaky}; a modification $\mathcal{D} \circ \mathcal{E}_A$ with deviation at most $d_{\text{drift}}$ on $\mathcal{S}$ is a \emph{drift event}, \emph{within-tolerance} if $d_{\text{drift}} \leq \varepsilon_A$. A within-tolerance record attests only that the deviation stayed inside $\varepsilon_A$, whether from drift, shot noise, or a substitution small enough to hide there.

\subsection{Detection Guarantee and the Adversary's Blind Set}
\label{sec:fw-detection}

The verifier measures single-qubit Pauli observables on reference inputs it prepares. The diamond norm is the wrong yardstick here: attaining it presumes an input entangled with a reference system and a joint measurement on output and reference, neither of which the protocol performs. The right quantity is the largest trace-norm shift of any single-qubit \emph{output marginal} over the reference inputs.

\begin{definition}[Marginal separation]
\label{def:marginal-sep}
For CPTP channels $\mathcal{E}_A, \mathcal{E}_B$ with $n$-qubit output and reference input set $\mathcal{S}$, write $\sigma_i(\rho) := \operatorname{Tr}_{\neq i}[\mathcal{E}_A(\rho)] - \operatorname{Tr}_{\neq i}[\mathcal{E}_B(\rho)]$. The marginal separation is
\begin{equation}
\delta_{\mathrm{loc}}(\mathcal{E}_A, \mathcal{E}_B; \mathcal{S}) := \sup_{\rho \in \mathcal{S}} \max_{1 \leq i \leq n} \big\| \sigma_i(\rho) \big\|_1 .
\label{eq:marginal-sep}
\end{equation}
\end{definition}

\noindent
Monotonicity of the trace norm under partial trace, then the diamond-norm definition, give the chain
\begin{equation}
\delta_{\mathrm{loc}}(\mathcal{E}_A, \mathcal{E}_B; \mathcal{S})
\;\leq\; \sup_{\rho \in \mathcal{S}} \| \mathcal{E}_A(\rho) - \mathcal{E}_B(\rho) \|_1
\;\leq\; \| \mathcal{E}_A - \mathcal{E}_B \|_\diamond .
\label{eq:norm-chain}
\end{equation}
Both inequalities can be strict and the gap total: if the declared channel outputs $\ket{\Phi^+}$, the substitution $Z \otimes I$ yields the orthogonal $\ket{\Phi^-}$, so $\delta_\diamond = 2$, yet every single-qubit marginal remains $I/2$ and $\delta_{\mathrm{loc}} = 0$.

\begin{theorem}[Two-sided detection bound]
\label{thm:detection}
Let $\mathcal{O}_A^{\text{local}} = \{X_i, Y_i, Z_i\}_{i=1}^{n}$. For all CPTP $\mathcal{E}_A, \mathcal{E}_B$ and every reference set $\mathcal{S}$,
\begin{equation}
\frac{1}{\sqrt{3}} \, \delta_{\mathrm{loc}}(\mathcal{E}_A, \mathcal{E}_B; \mathcal{S})
\;\leq\; \sup_{\rho \in \mathcal{S}} \Delta_{\mathcal{O}_A, \rho}(\mathcal{E}_A, \mathcal{E}_B)
\;\leq\; \delta_{\mathrm{loc}}(\mathcal{E}_A, \mathcal{E}_B; \mathcal{S}),
\label{eq:detection-bound}
\end{equation}
and both constants are tight. Hence any substitution with $\delta_{\mathrm{loc}} > \sqrt{3}\,\varepsilon_A$ violates the contract on some observable and some reference input, while any with $\delta_{\mathrm{loc}} \leq \varepsilon_A$ satisfies it everywhere. The constant $\sqrt{3}$ is independent of $n$.
\end{theorem}

\noindent\emph{Proof.} Appendix~\ref{app:proofs}. Each marginal difference is a traceless Hermitian operator on a single qubit, so the Bloch decomposition applies regardless of $n$; Cauchy--Schwarz gives $\sqrt{3}$, saturated at $\sigma^\ast \propto X+Y+Z$. No dimension-dependent conversion enters, in contrast to bounds routed through the diamond norm.\smallskip

\begin{proposition}[Exact blind set for local-unitary substitutions]
\label{prop:local-ic}
Let $\mathcal{E}_B = \mathrm{Ad}_U \circ \mathcal{E}_A$ with $U = U_1 \otimes \cdots \otimes U_n$ local, and write $\rho_i(\rho) := \operatorname{Tr}_{\neq i}[\mathcal{E}_A(\rho)]$. Then
\[
\delta_{\mathrm{loc}}(\mathcal{E}_A, \mathcal{E}_B; \mathcal{S}) > 0
\iff
\exists\, i,\; \exists\, \rho \in \mathcal{S}: \;\; [\,U_i,\, \rho_i(\rho)\,] \neq 0 .
\]
The \emph{blind set} of $(\mathcal{O}_A^{\text{local}}, \mathcal{S})$ is therefore exactly the local unitaries commuting with every output marginal of $\mathcal{E}_A$ on $\mathcal{S}$.
\end{proposition}

\noindent\emph{Proof.} Appendix~\ref{app:proofs}. The key identity is $\operatorname{Tr}_{\neq i}[U X U^\dagger] = U_i \operatorname{Tr}_{\neq i}[X] U_i^\dagger$ for local $U$, which reduces detectability to a commutation test on the output marginals.\smallskip

\paragraph{The blind set, operationally.}
Proposition~\ref{prop:local-ic} turns completeness into a check both sides can run, and two extremes bracket it. If the declared channel outputs a maximally entangled state on some $\rho \in \mathcal{S}$, every marginal is $I/2$ and commutes with everything, so \emph{every} non-trivial local unitary is blind there. If instead each qubit's marginals $\{\rho_i(\rho) : \rho \in \mathcal{S}\}$ contain two operators with no common eigenbasis, only global phases commute with all of them and the blind set collapses to the identity. Hence the configuration rule: \emph{choose reference inputs whose output marginals are non-degenerate and, across $\mathcal{S}$, non-commuting on every qubit}. The $S$-gate attack of Section~\ref{sec:experiment} lies in the detectable class, and the weak $\{Z_1Z_2\}$ family is blind to it by design. Substitutions that move only correlations need correlation observables, costed in Appendix~\ref{app:scope}.

\paragraph{An adaptive adversary, and the counter.}
Since $\mathcal{S}$ is known to $\mathcal{A}$, a rational attacker solves the commutation condition instead of picking an $S$-gate. If the challenge inputs drive every marginal onto one axis, every rotation about it is blind: $\mathcal{A}$ inserts one with diamond separation near $2$ and the fingerprint does not move. The counter is a better challenge set rather than a larger family: by Proposition~\ref{prop:local-ic}, two marginals per qubit with no common eigenbasis collapse the blind set to global phases, so drawing $\rho_{\text{ref}}$ per round from such a pool defeats the adaptive strategy at no extra shot cost.

\subsection{Soundness and Composition}
\label{sec:fw-soundness}

Detection needs a converse: benign proximity must not trigger halts. If $\| \mathcal{E}_A - \mathcal{E}_B \|_\diamond \leq d$ then $| \operatorname{Tr}(O \mathcal{E}_A(\rho)) - \operatorname{Tr}(O \mathcal{E}_B(\rho)) | \leq d \|O\|$ for every $O$ and $\rho$. This follows from H\"older plus the diamond-norm definition (Theorem~\ref{thm:soundness}, Appendix~\ref{app:proofs}). In words: for Pauli observables, channel drift of diamond magnitude $d$ moves no monitored expectation by more than $d$. Note the direction. The diamond norm is a legitimate \emph{upper} bound, dominating every restricted metric by \eqref{eq:norm-chain}; detection \emph{lower} bounds are stated in $\delta_{\mathrm{loc}}$. The two are never interconverted, since converting observable-level agreement into a channel-level bound is exactly what a restricted family cannot do, which is the content of the blind set. At tolerance $\varepsilon_A > 0$ the certified blind set enlarges to $\{\delta_{\mathrm{loc}} \leq \varepsilon_A\}$, and $\delta_{\mathrm{loc}} \in (\varepsilon_A, \sqrt{3}\,\varepsilon_A]$ is a grey zone, neither detection nor acceptance certified.

Stages compose. Take $L$ sequential stages with compatible domains and $\| \mathcal{E}_A^{(\ell)} - \mathcal{E}_B^{(\ell)} \|_\diamond \leq \varepsilon_\ell$. Then $\| \mathcal{E}_A - \mathcal{E}_B \|_\diamond \leq \sum_\ell \varepsilon_\ell$, and hence $\Delta_{\mathcal{O}_A, \rho} \leq (\sum_\ell \varepsilon_\ell)\max_O \|O\|$ (Theorem~\ref{thm:composition}, Appendix~\ref{app:proofs}). The per-stage hypothesis is deliberately a diamond bound, obtainable from process-level characterisation of individual stages. Classical pre- and post-processing falls outside this composition and is covered by the hash-chain layer of $\sigma_A$.

\subsection{Finite-Shot Budget with Two-Sided Error Control}
\label{sec:fw-complexity}

Exact expectations are unavailable, so the verifier estimates them from shots. A budget rule bounding only missed detections is not a defence: a check that flags honest channels as often as substituted ones has crossed a threshold for the wrong reason. We therefore require \emph{informative detection} at budget $N$: a true-positive rate $\mathrm{TPR}(N) \geq 1-\eta$ and a false-positive rate $\mathrm{FPR}(N) \leq \eta$.

\begin{theorem}[Measurement budget]
\label{thm:budget}
Let $|\mathcal{O}_A| = k$ with $\|O\| \leq B$ (Pauli: $B=1$), let $\delta > \sqrt{3}\,\varepsilon_A$ be the smallest marginal separation on the round's reference input that is to be detected, and set $\gamma := \delta/\sqrt{3} - \varepsilon_A > 0$. Assume an exact reference fingerprint and an honest channel equal to the declared one. Then at confidence $1-\eta$,
\begin{equation}
N \;\geq\; \frac{2 B^2 \, k \, \log(2 k / \eta)}{\min(\gamma, \varepsilon_A)^2}
\label{eq:budget-precomputed}
\end{equation}
suffices for informative detection, with shots spread uniformly over $\mathcal{O}_A$.
\end{theorem}

\noindent\emph{Proof.} Appendix~\ref{app:proofs}. Hoeffding with $t=\gamma$ union-bounded over $k$ controls misses; with $t = \varepsilon_A$ it controls false alarms; $\min(\gamma,\varepsilon_A)$ satisfies both.\smallskip

\begin{remark}[False-alarm margin under drift]
\label{rem:fpr-drift}
If the honest channel has drifted benignly with deviations bounded by $d < \varepsilon_A$, the false-alarm argument runs with margin $\varepsilon_A - d$. Hence the calibration below places $\varepsilon_A$ strictly above the typical drift. At the deployed values ($\varepsilon_A = 0.15$, $d = 0.067$) the margin is $0.083$, requiring $\approx 1{,}600$ shots per observable against a precomputed reference and four times that against a sampled one.
\end{remark}

When the reference is itself sampled, pairing shots and applying Hoeffding to the differences $Z_j \in [-2B,2B]$ replaces $2B^2$ by $8B^2$ (Corollary~\ref{cor:budget-sampled}, Appendix~\ref{app:proofs}). At the deployed parameters ($k=6$, $B=1$, $\delta = 0.5$, $\varepsilon_A = 0.15$, $\eta = 0.05$, so $\gamma \approx 0.139$) this gives $N \approx 3{,}420$ precomputed and $N = 13{,}680$ sampled; we run at the conservative value. For $k = 4^n - 1$ the linear dependence on $k$ dominates; classical shadows \cite{huang2020predicting} or adaptive shot allocation are the routes to tighter constants.

\subsection{Drift Mode and Tolerance Calibration}
\label{sec:fw-driftcor}

\begin{corollary}[Drift detection]
\label{cor:drift}
For a drift event $\mathcal{E}_B = \mathcal{D} \circ \mathcal{E}_A$, define the \emph{effective drift magnitude} $d_{\mathrm{loc}} := \delta_{\mathrm{loc}}(\mathcal{E}_A, \mathcal{D} \circ \mathcal{E}_A; \mathcal{S})$. Then
\begin{equation}
\frac{d_{\mathrm{loc}}}{\sqrt{3}} \;\leq\; \sup_{\rho \in \mathcal{S}} \Delta_{\mathcal{O}_A, \rho}(\mathcal{E}_A, \mathcal{E}_B) \;\leq\; d_{\mathrm{loc}} \;\leq\; \| \mathrm{id} - \mathcal{D} \|_\diamond,
\label{eq:drift-bound}
\end{equation}
so the event is detected when $d_{\mathrm{loc}} > \sqrt{3}\,\varepsilon_A$ and absorbed when $d_{\mathrm{loc}} \leq \varepsilon_A$.
\end{corollary}

\noindent\emph{Proof.} Appendix~\ref{app:proofs}.\smallskip

The relevant parameter is drift \emph{as seen through the declared channel on the verifier's inputs}, not the size of the perturbation. A $Z$ gate on a channel whose outputs are $Z$-eigenstates has $\|\mathrm{id}-\mathcal{D}\|_\diamond = 2$ but $d_{\mathrm{loc}} = 0$: nothing monitored moves, so no threshold follows from $\|\mathrm{id}-\mathcal{D}\|_\diamond$ alone.

Calibration has three steps. The operator (i) estimates typical drift $d_{\text{drift}}^{\text{typ}}$ from historical fingerprints, (ii) fixes the minimum separation $\delta_{\text{adv}}^{\min}$ to detect, and (iii) sets $\varepsilon_A \in (\, d_{\text{drift}}^{\text{typ}}, \; \delta_{\text{adv}}^{\min}/\sqrt{3} \,)$. The lower end keeps $\varepsilon_A$ above the drift floor, so a false-alarm margin remains (Remark~\ref{rem:fpr-drift}). The upper end keeps it below the point where Theorem~\ref{thm:budget}'s budget diverges. The three error sources enter separately: shot noise is absorbed by the budget $N$, drift fixes the interval's lower end, the target separation its upper end. An empty interval is a diagnostic: hardware drift exceeds the detection capability demanded of it, and the deployment must move to quieter hardware, accept a larger $\delta_{\text{adv}}^{\min}$, or restructure the family (Appendix~\ref{app:scope}).

\subsection{Verification Protocol}
\label{sec:fw-algorithm}

Algorithm~\ref{alg:verification} assembles the components. Each round measures the \emph{full} family so that the per-round guarantee of Theorem~\ref{thm:budget} applies as stated.

\begin{algorithm}[tb]
\scriptsize
\caption{Dual-mode behavioural fingerprinting verification.}
\label{alg:verification}
\begin{algorithmic}[1]
\Require spec $\sigma_A$, declared channel $\mathcal{E}_A$, reference set $\mathcal{S}$, confidence $1-\eta$, target separation $\delta$
\State $N \gets$ budget from Theorem~\ref{thm:budget}, or Corollary~\ref{cor:budget-sampled} if $\mathrm{Fp}_A$ is sampled
\For{round $t = 1, 2, \ldots, T$}
    \State draw $\rho_t \in \mathcal{S}$;\quad $\mathrm{Fp}_A \gets \mathrm{Fp}_{\mathcal{O}_A, \rho_t}(\mathcal{E}_A)$
    \State measure every $O \in \mathcal{O}_A$ (randomised order) with $n_O = \lceil N/|\mathcal{O}_A| \rceil$ shots on $\mathcal{E}(\rho_t)$, obtaining $\hat{\mu}_{t,O}$
    \State $\Delta_t \gets \max_{O \in \mathcal{O}_A} | \hat{\mu}_{t,O} - \mathrm{Fp}_A[O] |$
    \If{$\Delta_t > \varepsilon_A$} \Comment{integrity violation}
        \State commit $(t, \hat{\mu}_{t,\cdot}, \Delta_t)$ as beyond-tolerance; \Return \textsc{Halt}
    \Else \Comment{within tolerance}
        \State commit $(t, \hat{\mu}_{t,\cdot}, \Delta_t)$ as within-tolerance
    \EndIf
\EndFor
\State \Return \textsc{Accept}
\end{algorithmic}
\end{algorithm}

\noindent{\footnotesize\sffamily Implemented in \texttt{src/verification.py:run\_verifier}; see footnote~\ref{fn:repo}. Appendix~\ref{app:discussion} sketches integration with Qiskit, PennyLane and TensorFlow Quantum.}

\paragraph{Guarantees and accounting.}
With probability at least $1-\eta$, each round detects any substitution whose marginal separation on that round's input reaches $\delta$. Separations realised only on other inputs of $\mathcal{S}$ are caught in the rounds that use them, so enlarging $\mathcal{S}$ shrinks the blind set. The false-positive guarantee is per round: over $T$ rounds the probability of at least one spurious halt is at most $1-(1-\eta)^T \approx T\eta$, so an end-to-end budget $\eta_{\text{tot}}$ calls for $\eta = \eta_{\text{tot}}/T$, entering \eqref{eq:budget-precomputed} only logarithmically. Both kinds of record enter the audit trail anchored by $H_{\text{spec}}$ and are hash-chained, so the trail is \emph{tamper-evident}: rewriting a committed outcome breaks the chain and is caught on replay. It does not attribute records to an operator (Section~\ref{sec:limitations}), but supplies post hoc evidence for audit, dispute, or replication.

\section{Evaluation on IBM Heron r2 Hardware}
\label{sec:experiment}

We evaluate the defence end-to-end against a live channel-substitution attack on \texttt{ibm\_fez}, a 156-qubit IBM Heron r2 processor, using a two-qubit quantum-kernel intrusion-detection classifier, an AI security workload where a silently degraded model is costly and hard to notice. Three experiments follow. Does the attack evade the weak contract and get caught by ours (\S\ref{sec:exp-detection})? Is the prescribed budget where detection turns informative (\S\ref{sec:exp-sample})? Does drift stay inside a tolerance that still catches it (\S\ref{sec:exp-drift})? Metadata, job IDs, and raw counts are archived (Appendix~\ref{app:runs}, footnote~\ref{fn:repo}).

\paragraph{Setup.}
The victim is a QSVM with kernel entries from the inversion-test circuit $\ket{\psi_{ij}} = U_\phi(x_j)^\dagger U_\phi(x_i)\ket{00}$, $U_\phi$ being the two-qubit \texttt{ZZFeatureMap} at two repetitions. The verifier's reference inputs are the fixed pair $(x_i, x_j) = ((0.4,1.2),(1.1,0.3))$ from a synthetic two-class intrusion-detection task \cite{kalinin2023qiids}. The attack inserts an $S$-gate on every qubit before measurement. The $S$ insertion commutes with computational-basis readout, so the customer's $Z$-basis accuracy statistics are preserved exactly and (S1) holds by construction. The qubit subgraph is chosen by the default transpiler pass at \texttt{optimization\_level=1}. It instantiates both sides of Proposition~\ref{prop:local-ic}: $S$ fixes $\langle Z \rangle$ but rotates $\langle X \rangle, \langle Y \rangle$ by $\pi/2$, so the substitution sits in the blind set of $\{Z_1Z_2\}$ and in the detectable class of the local family. The $S$ insertion is a controlled probe of the bound, not an economically motivated substitution, and it isolates the marginal-detectable class. The experiment tests whether the separation of Theorem~\ref{thm:detection} survives device noise, finite sampling, and transpilation at a deployable budget.

We compare two families. The \emph{weak} one is $\{Z_1 Z_2\}$, a single correlation string of the kind used in kernel sanity checks. The \emph{local} one is $\{X_1, Y_1, Z_1, X_2, Y_2, Z_2\}$, whose frame constant $C$ is $\sqrt{3}$ by Theorem~\ref{thm:detection}. Here $C$ is the factor relating marginal separation to the worst per-observable deviation. Parameters: target marginal separation $\delta = 0.5$, $\varepsilon_A = 0.15$, $1-\eta = 0.95$, $k = 6$, $B = 1$, so $\gamma \approx 0.139$; runs use the conservative sampled-reference budget $N = 13{,}680$ ($n_O = 2{,}280$). For the budget sweep (Experiment 2) the reference fingerprint is precomputed on the ideal simulator at $10^5$ shots per observable; Experiments 1 and 3 compare fingerprints measured on hardware within the same batched job.

\subsection{Experiment 1: Attack Detection}
\label{sec:exp-detection}

\begin{table}[tb]
\caption{Experiment 1 on \texttt{ibm\_fez}: per-observable deviation between the honest and substituted channels, $|\langle P\rangle_A - \langle P\rangle_B|$, at $n_O = 2{,}280$ shots. Tolerance $\varepsilon_A = 0.15$; deviations beyond it in bold.}
\label{tab:exp1-deviations}
\centering
\renewcommand{\arraystretch}{1.0}
\scriptsize
\begin{tabular}{lcc}
\toprule
\textbf{Observable} & $|\langle P\rangle_A - \langle P\rangle_B|$ & \textbf{Within $\varepsilon_A$?} \\
\midrule
$Z_1 Z_2$ (weak family) & 0.001 & yes \\
\midrule
$X_1$ (local family) & \textbf{0.381} & no \\
$Y_1$ & 0.105 & yes \\
$Z_1$ & 0.004 & yes \\
$X_2$ & \textbf{0.489} & no \\
$Y_2$ & \textbf{0.358} & no \\
$Z_2$ & 0.005 & yes \\
\bottomrule
\end{tabular}
\end{table}

\begin{figure}[tb]
\centering
\includegraphics[width=0.72\textwidth]{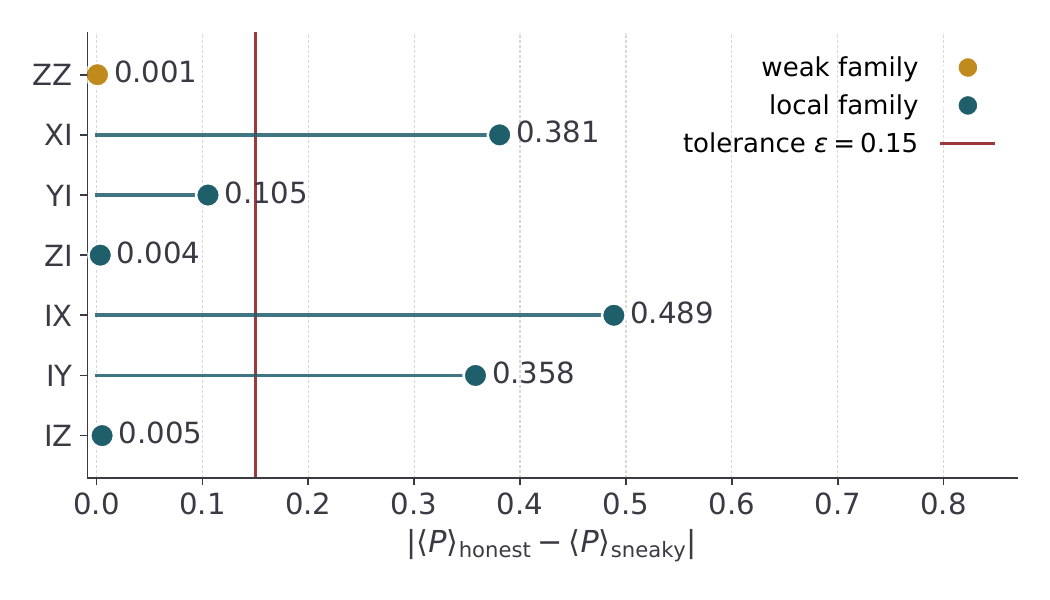}
\caption{Per-observable deviation between honest and sneaky channels on \texttt{ibm\_fez}, tolerance $\varepsilon_A = 0.15$ marked by the vertical line. Labels follow $P_i \equiv P \otimes I$ for $i = 1$ and $P_i \equiv I \otimes P$ for $i = 2$; the correlation observable is $Z_1 Z_2 \equiv Z \otimes Z$. Amber is the weak family, teal the local family. The weak contract accepts and the local contract halts, as Theorem~\ref{thm:detection} and Proposition~\ref{prop:local-ic} predict. Reproduced by \texttt{experiments/experiment1\_detection.py} with \texttt{analysis/plot\_figures.py:fig\_detection\_bars} (footnote~\ref{fn:repo}).}
\label{fig:exp1}
\end{figure}

The attack beats the weak contract and fails against ours (Table~\ref{tab:exp1-deviations} and Figure~\ref{fig:exp1}). Under $\{Z_1Z_2\}$ it deviates by $0.001$, at the noise floor, and the verifier accepts. Under the local family it produces $0.381$ on $X_1$, $0.489$ on $X_2$ and $0.358$ on $Y_2$, and the verifier halts at $3.3\times$ tolerance. Three of the six local observables independently exceed the tolerance, giving redundancy against single-observable estimation error.

The two-sided bound brackets the observation quantitatively. The second-qubit deviations give a measured marginal separation of $\|\sigma_2\|_1 = (0.489^2 + 0.358^2 + 0.005^2)^{1/2} \approx 0.606$. For that value, \eqref{eq:detection-bound} predicts a worst per-observable deviation in $[0.606/\sqrt{3},\, 0.606] = [0.350,\, 0.606]$. The observed $0.489$ falls inside. The measured separation also exceeds the deployment's target $\delta = 0.5$, confirming that the parameter regime assumed in calibration is the one realised on hardware.

\subsection{Experiment 2: Where Detection Becomes Informative}
\label{sec:exp-sample}

\begin{table}[tb]
\caption{Empirical TPR and FPR as a function of measurement budget, 20 trials per setting, weakened sneaky construction ($R_Z(\pi/6)$ insertion). The Corollary~\ref{cor:budget-sampled} (sampled-reference) budget corresponds to target $1 - \eta = 0.95$ TPR with FPR at most $\eta = 0.05$ for a substitution at the target separation $\delta$; this sweep's attack sits below it.}
\label{tab:exp2-tprfpr}
\centering
\renewcommand{\arraystretch}{1.0}
\scriptsize
\begin{tabular}{lccc}
\toprule
\textbf{Shot budget} & $n_O$ & \textbf{TPR (sneaky)} & \textbf{FPR (honest)} \\
\midrule
$N$ (Cor.~\ref{cor:budget-sampled}, conservative) & 2{,}280 & 1.00 & 0.00 \\
$N/10$ & 228 & 0.95 & 0.15 \\
$N/100$ & 22 & 1.00 & 1.00 \\
\bottomrule
\end{tabular}
\end{table}

\begin{figure}[tb]
\centering
\includegraphics[width=0.65\textwidth]{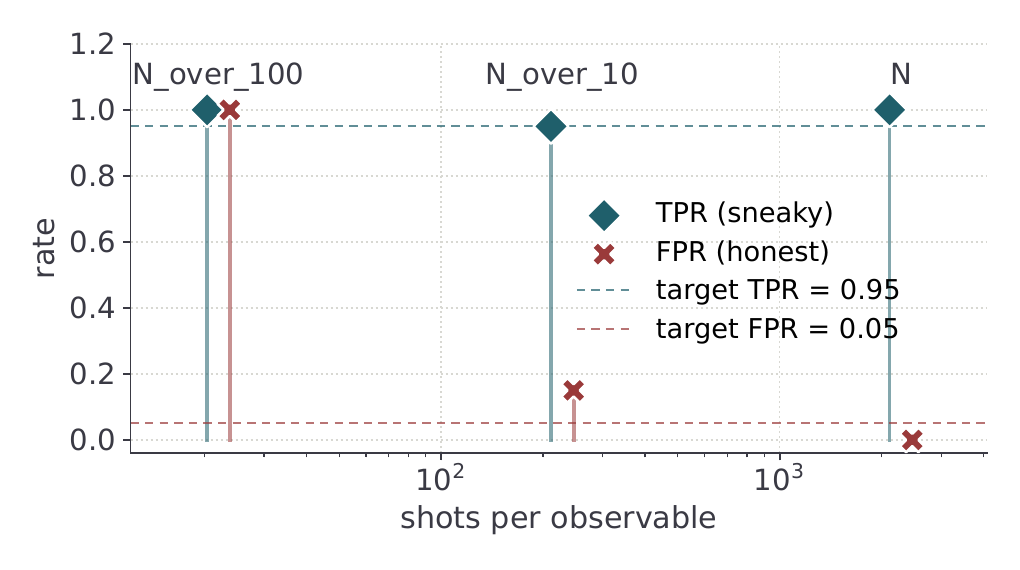}
\caption{Empirical TPR and FPR against shot budget per observable, weakened attack ($R_Z(\pi/6)$ insertion), 20 trials per setting on the Aer simulator with the \texttt{ibm\_fez} noise model. Dashed lines mark the target rates ($1-\eta = 0.95$ for TPR, $\eta = 0.05$ for FPR). The prescribed budget meets both; reduced budgets fail the FPR target, so detection there is uninformative. Reproduced by \texttt{experiments/experiment2\_sample.py} with \texttt{analysis/plot\_figures.py:fig\_sample\_curve} (footnote~\ref{fn:repo}).}
\label{fig:exp2}
\end{figure}

To probe the budget rule beyond its comfortable regime, the sweep uses a \emph{weakened} attack, $R_Z(\pi/6)$ instead of $S$. Its true maximum deviation $0.259$ sits below the certified threshold $\delta/\sqrt{3} \approx 0.2887$, so the sweep probes the budget beyond its guarantee region. Honest trials run on the noisy simulator against the noiseless ideal reference, so the measured FPR includes noise bias and is conservative. Each budget is repeated over 20 trials on an Aer simulator with the calibration-derived \texttt{ibm\_fez} noise model, since the repetitions make a hardware sweep impractical on queue time.

At the prescribed budget the verifier is informative (Table~\ref{tab:exp2-tprfpr} and Figure~\ref{fig:exp2}): $\mathrm{TPR} = 1.00$, $\mathrm{FPR} = 0.00$. At a tenth, TPR stays high ($0.95$) but FPR climbs to $0.15$, three times the nominal $\eta$, so the apparent detection is partly shot noise. At a hundredth both rates hit $1.00$: every run flags, honest or not. This is the failure a one-sided rule licenses, and why Theorem~\ref{thm:budget} controls both directions.

\subsection{Experiment 3: Drift Absorption}
\label{sec:exp-drift}

\begin{table}[tb]
\caption{Experiment 3 on \texttt{ibm\_fez}: maximum per-observable drift of the honest channel between timepoints, local family.}
\label{tab:exp3-drift}
\centering
\renewcommand{\arraystretch}{1.15}
\footnotesize
\begin{tabular}{lccc}
\toprule
\textbf{Timepoint pair} & $t_1 \to t_2$ & $t_1 \to t_3$ & $t_2 \to t_3$ \\
\midrule
Max per-observable drift & 0.067 & 0.067 & 0.046 \\
\bottomrule
\end{tabular}
\end{table}

\begin{figure}[tb]
\centering
\includegraphics[width=0.76\textwidth]{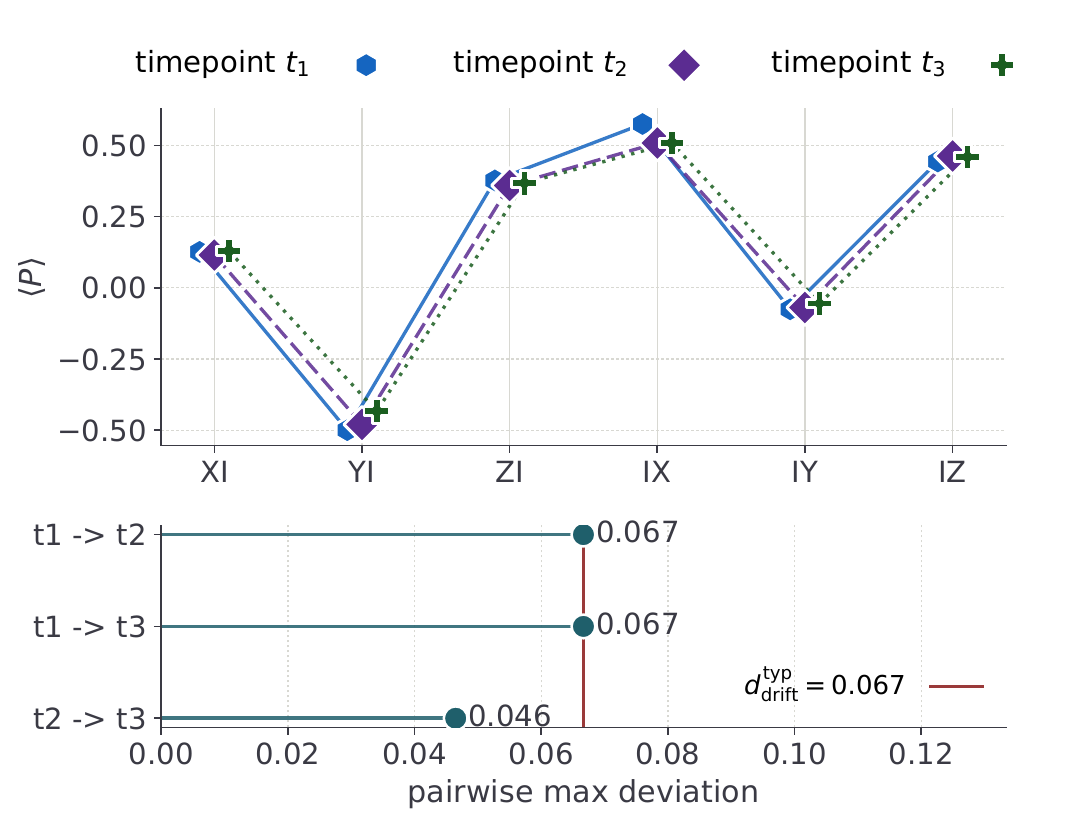}
\caption{Drift across three timepoints within one batched submission on \texttt{ibm\_fez}. Top: the honest fingerprint at each timepoint, distinct markers and line styles; the traces nearly coincide. Bottom: the pairwise maximum deviations of Table~\ref{tab:exp3-drift}, with $d_{\text{drift}}^{\text{typ}} = 0.067$ marked by the vertical line, well below $\varepsilon_A = 0.15$. Reproduced by \texttt{experiments/experiment3\_drift.py} with \texttt{analysis/plot\_figures.py:fig\_drift\_panel} (footnote~\ref{fn:repo}).}
\label{fig:exp3}
\end{figure}

Fingerprinting the honest channel at three timepoints (Table~\ref{tab:exp3-drift} and Figure~\ref{fig:exp3}) gives an observed deviation floor $d_{\text{drift}}^{\text{typ}} = 0.067$, an upper estimate of the drift magnitude of Corollary~\ref{cor:drift}. The calibration procedure then admits
$\varepsilon_A \in (\,0.067,\; 0.5/\sqrt{3}\,) = (0.067,\, 0.2887)$,
a non-empty interval: on this backend, attack detection and drift absorption coexist within a job rather than trade off. The deployed $\varepsilon_A = 0.15$ sits inside it, with a detection margin of $\gamma \approx 0.139$ and a false-alarm margin of $0.083$ (Remark~\ref{rem:fpr-drift}). Certifying that second margin at the sampled-reference constant needs $\approx 6{,}400$ shots per observable, above the $2{,}280$ run here. Every deviation observed stayed within tolerance and was logged, not halted, exercising the monitoring mode end-to-end.

\paragraph{Cost of the defence.}
The candidate side of a round is $13{,}680$ shots. The archived job measured candidate and reference together, $31{,}920$ shots over $14$ circuits. It took $30.7$\,s wall clock on \texttt{ibm\_fez}, of which $\approx 15$\,s QPU (Appendix~\ref{app:runs}). Classical post-processing, six comparisons and a hash-chain append, takes microseconds. Cost scales as $k \cdot n_O$, so it is governed by the observable family rather than by the protected pipeline.

\paragraph{Threats to validity.}
The evidence is one pipeline, one backend, single batched jobs (Appendix~\ref{app:runs}). Drift evidence is within-batch: the three timepoints sit inside one $30$-second job, so $d_{\text{drift}}^{\text{typ}} = 0.067$ is a shot-noise-dominated fluctuation floor, an upper bound on fast drift only. Across the two archived jobs, $31$ minutes apart and independently transpiled without a pinned layout, honest fingerprints differ by up to $0.27$, concentrated on the first qubit ($0.24$ to $0.27$, against $0.11$ to $0.14$ on the second), which is consistent with the unpinned layout rather than with uniform decoherence; the demonstration therefore compares two channels inside one job, since a deployed verifier cannot obtain an honest hardware reference from the untrusted provider. A simulated reference sits $0.21$ and $0.42$ from honest hardware across the two jobs, so at this noise level the deployable anchor is a trusted prior run with a pinned layout; the non-empty interval is a within-job statement, and a deployment recalibrates on its own window. Deviations are point estimates. The two-qubit instance exercises encoder, ansatz, measurement, feature map, sampling, and drift, but does not probe the regime where family size and crosstalk grow (Appendix~\ref{app:scope}). Scaling along all three axes is the direct next step.

\section{Limitations}
\label{sec:limitations}

\paragraph{Input-set dependence.}
The guarantees are exact, but exact \emph{about the deployed pair} $(\mathcal{O}_A, \mathcal{S})$. A deployment locking one canonical reference input has the largest blind set, which a white-box adversary can target; Proposition~\ref{prop:local-ic} makes that boundary checkable, and therefore a configuration lever. No fixed-input, local-measurement verifier can offer an input-independent guarantee against all channel-level substitutions, and stating the boundary is preferable to claiming past it. \emph{Correlation-only substitutions.} Substitutions moving joint correlations while fixing every marginal are invisible to the local family on \emph{any} $\mathcal{S}$: entangling-gate insertion, correlated crosstalk, subspace-preserving permutations. Detecting them needs correlation observables, which raise the frame constant and shrink the margin. For kernel pipelines the target quantity is itself a correlation: the inversion-test kernel is $K = \tfrac{1}{4}(1 + \langle Z_1 \rangle + \langle Z_2 \rangle + \langle Z_1 Z_2 \rangle)$, so a correlation-only substitution shifts every kernel entry while $\delta_{\mathrm{loc}} = 0$. A QSVM deployment should therefore add $Z_1 Z_2$ to the local family. That is the kernel tier of Appendix~\ref{app:scope}, with $k = 7$, $C \approx 1.83$ and $N \approx 20{,}700$ here, costing $1.5\times$ Tier~1 against the $5.3\times$ of Tier~2. We report the local family because our probe lies in the marginal-detectable class. The three diagonal correlations move $C$ from $\sqrt{3}$ to $\approx 2.21$, multiplying the budget by $\approx 5.3$ at $\varepsilon_A = 0.15$; the full 15-string family pushes $C \approx 3.73$, making $\gamma$ negative. Appendix~\ref{app:scope} quantifies the tiers.

\paragraph{Reference-fingerprint circularity.}
A natural objection is that in QML the correct output is what is unknown, so a verifier that could compute it would not need the pipeline. The contract does not verify the production answer. It verifies \emph{channel identity on designated reference inputs}, and that fingerprint must be computable. Classical simulation can do it at verification scale where its device bias sits inside $\varepsilon_A$; on this backend it does not (Section~\ref{sec:exp-drift}), so the anchor here is a trusted prior run. Integrity on $\mathcal{S}$ constrains production behaviour to the extent that $\mathcal{S}$ exercises the same mechanisms. Where reference fingerprints exceed classical simulability, their generation becomes a genuine limitation, to be replaced by trusted prior runs (Appendix~\ref{app:anchor}), cross-device consistency, or interactive verification \cite{mahadev2018classical,fitzsimons2017unconditionally}. \emph{Behavioural, not cryptographic, authentication.} The contract authenticates \emph{executions} against a specification by statistical indistinguishability of measured observables; it does not authenticate \emph{parties}. No key, signature, or certificate binds a record to an operator, so the trail attests what was measured but not who produced it. An adversary forging the classical transport of outcomes, or controlling the verifier, is outside the model; defending those surfaces needs signed results, hardware roots of trust, or interactive verification, which compose with rather than replace this contract.

\paragraph{Scope of the contract.}
The contract certifies channel identity, not model quality: a faithfully executed pipeline with poor parameters stays poor, and a substitution that improves accuracy is still halted. Pipelines emitting only classical post-processed decisions need the contract lifted to an intermediate measurement layer. Sample complexity grows linearly in family size, dominating for $k = 4^n-1$.

\section{Conclusion}
\label{sec:conclusion}

Delegating quantum computation to uncontrolled hardware creates an execution-integrity gap, and this paper shows it is attackable: a channel-substitution attack reproduces the customer's accuracy tests and single-observable checks while running a different channel.

QML-PipeGuard closes it with a runtime contract needing no trusted hardware, no interactive proof, and no change to the customer's circuits. Four results, all established here, support it. Theorem~\ref{thm:detection} bounds the contract deviation on both sides by the marginal separation, with $\sqrt{3}$ tight for every qubit count. Proposition~\ref{prop:local-ic} settles which local-unitary substitutions the contract can see. Theorem~\ref{thm:budget} fixes a budget bounding false accepts and false alarms at once, and Corollary~\ref{cor:drift} measures drift through the declared channel, so one tolerance covers both. On \texttt{ibm\_fez} the separation is demonstrable today. The attack evades the weak contract at $0.001$ and is halted at $0.489$ against tolerance $0.15$, the $13{,}680$ shots fit one batched job of about half a minute, and the within-job interval $(0.067, 0.2887)$ is non-empty, so detection and drift absorption coexist there. The comparison is between two channels in one job, not against a trusted reference. Every round, within tolerance or halted, enters a tamper-evident trail anchored by the specification hash, showing auditors what the provider executed. As AI systems delegate stages to quantum hardware under audit, such evidence becomes a requirement rather than a convenience. \emph{Open problems remain.} Tighter frame constants and adaptive shot allocation would cut the budget. A blind-set characterisation for entangling and non-unitary substitutions would close the coverage gap. Reference generation beyond classical simulability is unresolved. A witnessed baseline can serve as the reference (Section~\ref{sec:limitations}).

\newpage
\appendix
\renewcommand{\theHsection}{\Alph{section}}

\section{Proofs}
\label{app:proofs}

\subsection*{Proof of Theorem~\ref{thm:detection}}
\label{pf:thm:detection}
\begin{proof}
Fix $\rho \in \mathcal{S}$ and $i \in \{1, \dots, n\}$. Both marginals have unit trace, so $\sigma := \sigma_i(\rho)$ is a traceless Hermitian operator on $\mathbb{C}^2$ and admits the Bloch decomposition
\[
\sigma \;=\; a_X X + a_Y Y + a_Z Z,
\qquad
a_P \;=\; \tfrac{1}{2}\operatorname{Tr}(P \sigma) \;\; \text{for } P \in \{X, Y, Z\}.
\]
The eigenvalues of $\sigma$ are $\pm |\vec{a}|$ with $|\vec{a}| = \sqrt{a_X^2 + a_Y^2 + a_Z^2}$, hence $\| \sigma \|_1 = 2 |\vec{a}|$. Moreover, by the defining property of the partial trace,
\[
\operatorname{Tr}\!\big( P_i \, [\mathcal{E}_A(\rho) - \mathcal{E}_B(\rho)] \big)
\;=\;
\operatorname{Tr}\!\big( P \, \sigma_i(\rho) \big)
\;=\; 2 a_P
\qquad \text{for } P_i \in \{X_i, Y_i, Z_i\},
\]
so the per-observable deviations of the contract are exactly the quantities $|2 a_P|$ ranging over qubits and Pauli axes.

\emph{Upper bound.} Each coordinate of a vector is bounded by its Euclidean norm, so $\max_P |2 a_P| \leq 2|\vec{a}| = \|\sigma\|_1$. Taking the maximum over $i$ and the supremum over $\rho \in \mathcal{S}$ gives the right-hand inequality of \eqref{eq:detection-bound}; equality holds whenever $\sigma$ is proportional to a single Pauli operator.

\emph{Lower bound.} By Cauchy--Schwarz, $|\vec{a}| \leq \sqrt{3} \cdot \max_P |a_P|$, hence
\[
\max_{P} \big| \operatorname{Tr}(P \sigma) \big| \;=\; 2 \max_P |a_P| \;\geq\; \frac{2 |\vec{a}|}{\sqrt{3}} \;=\; \frac{\| \sigma \|_1}{\sqrt{3}},
\]
with equality if and only if $|a_X| = |a_Y| = |a_Z|$; the witness $\sigma^\ast \propto X + Y + Z$ saturates the bound and is realised by valid states, $\rho_\pm = \tfrac{1}{2}\bigl(I \pm (X+Y+Z)/\sqrt{3}\bigr)$ being density matrices whose difference is proportional to it, so the constant is tight. Therefore $\Delta_{\mathcal{O}_A, \rho}(\mathcal{E}_A, \mathcal{E}_B) \geq \max_i \| \sigma_i(\rho) \|_1 / \sqrt{3}$ for every $\rho$, and taking suprema over $\rho \in \mathcal{S}$ on both sides yields the left-hand inequality of \eqref{eq:detection-bound}.

The constant is independent of $n$. Every marginal difference $\sigma_i(\rho)$ lives on a single qubit, whatever the size of the joint output space, so no dimension-dependent conversion factor enters. Bounds routed through the diamond norm do carry one, because they require an entangled reference system that the protocol never prepares. The detection and absorption statements follow immediately from \eqref{eq:detection-bound} and the contract threshold $\varepsilon_A$.
\end{proof}

\subsection*{Proof of Proposition~\ref{prop:local-ic}}
\label{pf:prop:local-ic}
\begin{proof}
For any operator $X$ on the $n$-qubit space and any local unitary $U = \bigotimes_j U_j$, the partial trace over the qubits other than $i$ satisfies
\[
\operatorname{Tr}_{\neq i}\!\big[\, U X U^\dagger \,\big] \;=\; U_i \, \operatorname{Tr}_{\neq i}[X] \, U_i^\dagger,
\]
because conjugation by a unitary acting only on the traced-out factors leaves the partial trace invariant, while conjugation by $U_i$ on the retained factor commutes with tracing out the others. Applying this with $X = \mathcal{E}_A(\rho)$ gives, for every $\rho \in \mathcal{S}$ and every $i$,
\[
\sigma_i(\rho) \;=\; \rho_i(\rho) \;-\; U_i \, \rho_i(\rho) \, U_i^\dagger .
\]

($\Leftarrow$) Suppose $[U_i, \rho_i(\rho)] \neq 0$ for some $i, \rho$. For a unitary $U_i$, commutation with a Hermitian operator is equivalent to fixing it under conjugation, so $U_i \rho_i(\rho) U_i^\dagger \neq \rho_i(\rho)$ and hence $\sigma_i(\rho) \neq 0$, giving $\|\sigma_i(\rho)\|_1 > 0$ and $\delta_{\mathrm{loc}} > 0$. By the lower bound of Theorem~\ref{thm:detection}, some $P \in \{X_i, Y_i, Z_i\}$ then exhibits a strictly positive expectation gap on $\rho$: this is the informational completeness of the single-qubit Pauli family applied to the nonzero traceless operator $\sigma_i(\rho)$.

($\Rightarrow$) Conversely, suppose $[U_i, \rho_i(\rho)] = 0$ for every $i$ and every $\rho \in \mathcal{S}$. Then $\sigma_i(\rho) = 0$ identically, so $\delta_{\mathrm{loc}} = 0$, and by the upper bound of Theorem~\ref{thm:detection} every local-Pauli expectation of $\mathcal{E}_B$ coincides with that of $\mathcal{E}_A$ on all of $\mathcal{S}$: the substitution satisfies the contract at tolerance $0$ and is invisible to the verifier.
\end{proof}

\subsection*{Soundness}
\begin{theorem}[Soundness]
\label{thm:soundness}
Let $\mathcal{E}_A, \mathcal{E}_B$ be CPTP with $\| \mathcal{E}_A - \mathcal{E}_B \|_\diamond \leq d$. Then for every $O \in \mathcal{O}_A$ and every reference input $\rho$, $| \operatorname{Tr}(O \mathcal{E}_A(\rho)) - \operatorname{Tr}(O \mathcal{E}_B(\rho)) | \leq d \cdot \| O \|$.
\end{theorem}
\begin{proof}
Let $\Delta_\rho := \mathcal{E}_A(\rho) - \mathcal{E}_B(\rho)$. By the trace-norm/operator-norm duality \cite[Section 1.1.3]{watrous2018theory},
\[
\big| \operatorname{Tr}(O \Delta_\rho) \big| \;\leq\; \| O \| \cdot \| \Delta_\rho \|_1,
\]
which is the operator H\"older inequality on $1$- and $\infty$-Schatten norms. By the definition of the diamond norm as the maximum over all input states (and auxiliary systems) of the trace-norm distance between channel outputs \cite[Section 3.3]{watrous2018theory},
\[
\| \Delta_\rho \|_1 \;=\; \| (\mathcal{E}_A - \mathcal{E}_B)(\rho) \|_1 \;\leq\; \| \mathcal{E}_A - \mathcal{E}_B \|_\diamond \;\leq\; d.
\]
Combining gives $|\operatorname{Tr}(O \Delta_\rho)| \leq d \cdot \|O\|$, which is the soundness bound.
\end{proof}

\subsection*{Compositionality}
\begin{theorem}[Compositionality]
\label{thm:composition}
Let a pipeline have $L$ sequential stages with declared channels $\mathcal{E}_A^{(\ell)}$ and candidates $\mathcal{E}_B^{(\ell)}$, the output space of stage $\ell$ matching the input space of stage $\ell+1$ in both chains. If $\| \mathcal{E}_A^{(\ell)} - \mathcal{E}_B^{(\ell)} \|_\diamond \leq \varepsilon_\ell$ for each $\ell$, then the composed channels satisfy $\| \mathcal{E}_A - \mathcal{E}_B \|_\diamond \leq \sum_{\ell} \varepsilon_\ell$ and hence $\Delta_{\mathcal{O}_A, \rho}(\mathcal{E}_A, \mathcal{E}_B) \leq (\sum_{\ell} \varepsilon_\ell) \max_{O \in \mathcal{O}_A} \| O \|$ for every $\rho$.
\end{theorem}
\begin{proof}
Insert telescoping terms:
\[
\mathcal{E}_A - \mathcal{E}_B \;=\; \sum_{\ell=1}^{L} \mathcal{C}_\ell, \quad \text{where} \quad
\mathcal{C}_\ell \;=\; \mathcal{E}_B^{(L)} \circ \cdots \circ \mathcal{E}_B^{(\ell+1)} \circ \big(\mathcal{E}_A^{(\ell)} - \mathcal{E}_B^{(\ell)}\big) \circ \mathcal{E}_A^{(\ell-1)} \circ \cdots \circ \mathcal{E}_A^{(1)},
\]
which is well-defined by the domain-compatibility hypothesis. The diamond norm is subadditive under sums and contractive under pre- and post-composition with CPTP channels \cite[Section 3.3]{watrous2018theory}, so
\[
\| \mathcal{E}_A - \mathcal{E}_B \|_\diamond \;\leq\; \sum_{\ell=1}^L \| \mathcal{C}_\ell \|_\diamond \;\leq\; \sum_{\ell=1}^L \| \mathcal{E}_A^{(\ell)} - \mathcal{E}_B^{(\ell)} \|_\diamond \;\leq\; \sum_{\ell=1}^L \varepsilon_\ell.
\]
The observable-level bound follows by applying Theorem~\ref{thm:soundness} to the composed channels.
\end{proof}

\subsection*{Proof of Theorem~\ref{thm:budget}}
\label{pf:thm:budget}
\begin{proof}
The proof combines a per-observable Hoeffding concentration with a union bound across the family, once for each error direction, then solves for $N$.

\textbf{Step 1: Per-observable concentration.}
For a single observable $O \in \mathcal{O}_A$, each measurement shot yields a sample $X_i$ of a random variable with mean $\mu_O := \operatorname{Tr}(O \mathcal{E}(\rho))$ and bounded by $|X_i| \leq \| O \| \leq B$ (since each $X_i$ is an eigenvalue of $O$). We assume the standard measurement protocol in which the reference state $\rho$ is newly prepared before each shot, ensuring that the samples $X_1, \ldots, X_{n_O}$ are i.i.d.\ random variables on the interval $[-B, B]$. By Hoeffding's inequality \cite[Theorem 2]{hoeffding1963probability}, after $n_O$ shots the empirical mean $\hat{\mu}_O = \frac{1}{n_O} \sum_{i=1}^{n_O} X_i$ satisfies, for any $t > 0$,
\begin{equation}
\Pr\!\Big[ \big| \hat{\mu}_O - \mu_O \big| > t \Big] \;\leq\; 2 \exp\!\left( -\frac{n_O t^2}{2 B^2} \right),
\label{eq:hoeffding-single}
\end{equation}
where the factor $2$ accounts for two-sided deviations and the denominator $2B^2$ arises from the interval width $b - a = 2B$ in Hoeffding's general formula $\exp(-2n t^2 / (b-a)^2)$.

\textbf{Step 2: False negatives (missing a separated substitute).}
Suppose the executed channel is a substitution with $\delta_{\mathrm{loc}} \geq \delta$. By Theorem~\ref{thm:detection}, its true deviation from the reference exceeds $\delta / \sqrt{3} = \varepsilon_A + \gamma$ on some observable $O^\ast$, on that input. With $\mathrm{Fp}_A$ exact, the verifier misses the violation only if the empirical estimate of that observable errs by more than $\gamma$: setting $t = \gamma$ in \eqref{eq:hoeffding-single} and union-bounding over all $k$ observables (any of which may carry the violation),
\begin{equation}
\Pr[\text{miss}] \;\leq\; 2k \exp\!\left( -\frac{n_O \gamma^2}{2 B^2} \right).
\label{eq:union-bound}
\end{equation}

\textbf{Step 3: False positives (flagging the honest channel).}
Suppose the executed channel equals the declared one, so every true deviation is $0$. The verifier flags a violation only if some empirical estimate errs by more than $\varepsilon_A$: setting $t = \varepsilon_A$ in \eqref{eq:hoeffding-single} and union-bounding over the $k$ observables,
\begin{equation}
\Pr[\text{false alarm}] \;\leq\; 2k \exp\!\left( -\frac{n_O \varepsilon_A^2}{2 B^2} \right).
\label{eq:union-bound-fpr}
\end{equation}

\textbf{Step 4: Solve for $N$.}
Let $N$ be the total shot budget allocated uniformly, $n_O = N/k$. Requiring each of \eqref{eq:union-bound} and \eqref{eq:union-bound-fpr} to be at most $\eta$ gives, respectively, $n_O \geq 2B^2 \log(2k/\eta)/\gamma^2$ and $n_O \geq 2B^2 \log(2k/\eta)/\varepsilon_A^2$; both are implied by
\[
n_O \;\geq\; \frac{2 B^2 \log(2k / \eta)}{\min(\gamma, \varepsilon_A)^2},
\]
which is Equation~\eqref{eq:budget-precomputed} after multiplying by $k$.
\end{proof}

\subsection*{Sampled-Reference Budget}
\begin{corollary}[Sampled-reference budget]
\label{cor:budget-sampled}
If the reference fingerprint is itself estimated empirically at the same per-observable shot count, then under the hypotheses of Theorem~\ref{thm:budget} the budget $N \geq 8 B^2 k \log(2k/\eta) / \min(\gamma, \varepsilon_A)^2$ suffices for informative detection.
\end{corollary}
\begin{proof}
For each observable $O$, the verifier's decision statistic is the difference $\hat{D}_O := \hat{\mu}_O^{\text{cand}} - \hat{\mu}_O^{\text{ref}}$ of two independent empirical means. Pair the shots and consider the i.i.d.\ differences $Z_j := X_j^{\text{cand}} - X_j^{\text{ref}} \in [-2B, 2B]$, whose empirical mean is $\hat{D}_O$ and whose true mean is the true deviation $D_O$. Hoeffding's inequality on the interval of width $4B$ gives, for any $t > 0$,
\[
\Pr\!\big[ |\hat{D}_O - D_O| > t \big] \;\leq\; 2 \exp\!\left( -\frac{2 n_O t^2}{(4B)^2} \right) \;=\; 2 \exp\!\left( -\frac{n_O t^2}{8 B^2} \right).
\]
Repeating Steps 2--4 of Theorem~\ref{thm:budget} verbatim with this bound in place of \eqref{eq:hoeffding-single} (with $t = \gamma$ for the miss direction and $t = \varepsilon_A$ for the false-alarm direction, each union-bounded over $k$ observables) yields the $8B^2$ numerator and the same $\log(2k/\eta)$ factor, which is the sampled-reference bound. 
\end{proof}

\subsection*{Proof of Corollary~\ref{cor:drift}}
\label{pf:cor:drift}
\begin{proof}
The two-sided bound on the contract deviation is Theorem~\ref{thm:detection} applied to the pair $(\mathcal{E}_A, \mathcal{D} \circ \mathcal{E}_A)$, whose composition is CPTP since $\mathcal{D}$ is. For the rightmost inequality, the norm chain \eqref{eq:norm-chain} gives $d_{\mathrm{loc}} \leq \| \mathcal{E}_A - \mathcal{D} \circ \mathcal{E}_A \|_\diamond = \| (\mathrm{id} - \mathcal{D}) \circ \mathcal{E}_A \|_\diamond \leq \| \mathrm{id} - \mathcal{D} \|_\diamond$, using contractivity of the diamond norm under pre-composition with the CPTP channel $\mathcal{E}_A$ \cite[Section 3.3]{watrous2018theory}. The detection and absorption conditions follow from \eqref{eq:drift-bound} and the threshold $\varepsilon_A$.
\end{proof}

\section{Per-Class Substitution Forms}
\label{app:qml}

The threat model of Section~\ref{sec:threat-adv} applies to any QML channel. For each of the three target model classes, the substitution takes a specific operational form. In each case, the agreement conditions (S1) and (S2) hold at the level of classification decisions and on a weaker observable family, while the channel-level difference (S3) manifests in the full quantum output state:

\begin{itemize}
\item \textbf{VQC}: $\mathcal{E}_B$ corresponds to a different parameter vector $\theta_B \neq \theta_A$ such that the trained ansatz $W(\theta_B)$ produces (S1) the same classification decisions $\arg\max f_B(x) = \arg\max f_A(x)$ on the verification test set, and (S2) matching expectation values on the weaker observable family $\mathcal{O}_{\text{weak}}$. However, the output quantum state $W(\theta_B) \ket{\phi(x)}\!\bra{\phi(x)} W(\theta_B)^\dagger$ differs from $W(\theta_A) \ket{\phi(x)}\!\bra{\phi(x)} W(\theta_A)^\dagger$ in its entanglement structure or confidence distribution, which (S3) manifests as a channel-level difference; it is exposed by the contract exactly when it shifts an output marginal (Theorem~\ref{thm:detection}).

\item \textbf{QSVM}: $\mathcal{E}_B$ corresponds to a different encoding circuit $U_{\phi_B}$ such that (S1) the resulting kernel matrix $K_B(x_i, x_j)$ produces the same SVM classification decisions on the verification test set as $K_A(x_i, x_j)$, and (S2) agrees with $K_A$ on diagonal entries and weak observable measurements used during basic kernel sanity checks. However, the encoded quantum states $\ket{\phi_B(x)}$ differ from $\ket{\phi_A(x)}$ in observable structure beyond the weak subset, producing (S3) different off-diagonal kernel values or shifted output marginals that are exposed by the local Pauli family when Proposition~\ref{prop:local-ic} applies.

\item \textbf{QNN}: $\mathcal{E}_B$ corresponds to a different layer configuration $\{\theta_\ell^B\}$ such that (S1) the final-layer measurement of a designated observable yields the same classification decisions on the verification test set as under $\{\theta_\ell^A\}$, and (S2) weak observable expectations match within tolerance. However, the intermediate-layer states $(\mathcal{L}_{\theta_\ell^B} \circ \cdots \circ \mathcal{L}_{\theta_1^B} \circ \mathcal{E}_\phi)(\rho)$ differ from those under $\{\theta_\ell^A\}$, producing (S3) measurable differences that the local family exposes whenever an output marginal is shifted.
\end{itemize}

\section{Run Metadata}
\label{app:runs}
Tables~\ref{tab:exp1-metadata} and~\ref{tab:exp3-metadata} give the metadata of the two hardware jobs. The corresponding archives in the repository (footnote~\ref{fn:repo}) are \texttt{results/\ab{}qpu\_\ab{}runs/\ab{}experiment1\_\ab{}ibm\_\ab{}fez\_\ab{}20260522\_\ab{}120209.json} for the detection run, \texttt{results/\ab{}qpu\_\ab{}runs/\ab{}experiment3\_\ab{}drift\_\ab{}ibm\_\ab{}fez\_\ab{}20260522\_\ab{}123327.json} for the drift run, and \texttt{results/\ab{}qpu\_\ab{}runs/\ab{}experiment2\_\ab{}sample\_\ab{}fake\_\ab{}fez\_\ab{}20260522\_\ab{}131130.json} for the simulator sweep. Each holds the raw shot counts, the per-observable fingerprints, the audit trail, and the run parameters, so every table and figure can be regenerated without QPU access via \texttt{analysis/plot\_figures.py}.

\begin{table}[tb]
\caption{Run metadata for the detection experiment on \texttt{ibm\_fez}.}
\label{tab:exp1-metadata}
\centering
\renewcommand{\arraystretch}{1.05}
\footnotesize
\begin{tabular}{ll}
\toprule
Backend & \texttt{ibm\_fez} (IBM Heron r2, 156 qubits) \\
Job ID & \texttt{d884b8is46sc73f8v28g} \\
Date & 22 May 2026, 12:02 UTC \\
Plan & IBM Quantum Open Plan \\
Shots & 2{,}280 per circuit \\
Total circuits & 14 (honest + sneaky, 7-Pauli family) \\
Queue time & $\approx$ 15 seconds \\
QPU run time & $\approx$ 15 seconds \\
Total wall time & 30.7 seconds \\
\bottomrule
\end{tabular}
\end{table}

\begin{table}[tb]
\caption{Run metadata for the drift experiment on \texttt{ibm\_fez}.}
\label{tab:exp3-metadata}
\centering
\renewcommand{\arraystretch}{1.05}
\footnotesize
\begin{tabular}{ll}
\toprule
Backend & \texttt{ibm\_fez} (IBM Heron r2, 156 qubits) \\
Job ID & \texttt{d884pu2s46sc73f8vn0g} \\
Date & 22 May 2026, 12:33 UTC \\
Plan & IBM Quantum Open Plan \\
Shots & 2{,}280 per circuit \\
Total circuits & 18 (3 timepoints, 6-Pauli local family) \\
Queue time & $\approx$ 15 seconds \\
QPU run time & $\approx$ 15 seconds \\
Total wall time & 30.5 seconds \\
\bottomrule
\end{tabular}
\end{table}

\paragraph{What ran where, and why.}
All hardware runs used the IBM Quantum Open Plan, which allows ten minutes of QPU
time per 28-day window and does not support session mode: workloads run as single
jobs or as batches, and consecutive jobs may be separated by a recalibration. Each
experiment therefore had to fit into one batched job of about $30$ seconds.
Table~\ref{tab:where} shows what ran where. Detection (\S\ref{sec:exp-detection})
and drift absorption (\S\ref{sec:exp-drift}), the two claims the paper rests on,
are both hardware results on \texttt{ibm\_fez} with archived job IDs. Only the
budget sweep is simulated. It needs $3$ budgets $\times$ $20$ trials $\times$ $7$
circuits, that is $420$ separate executions, which queue time makes impractical
under this plan. It therefore runs on the \texttt{ibm\_fez} noise model derived
from the same backend's calibration data. Honest trials there are scored against a
noiseless reference, so the reported FPR is conservative.

The plan limited what we could measure, not what the method can do. The detection
job used about $15$\,s of QPU time (Table~\ref{tab:exp1-metadata}), so tracking
drift hourly for a day would take $24$ such jobs, some six minutes of QPU: most of
the plan's 28-day allowance of QPU time, which also had to cover the rest of our
work. The longer study needs a larger allowance rather than a different method,
and is left to future work.

\begin{table}[!ht]
\caption{Where each experiment was executed. Job IDs and archives as above.}
\label{tab:where}
\centering
\renewcommand{\arraystretch}{1.05}
\footnotesize
\begin{tabular}{llll}
\toprule
Experiment & Platform & Evidence & Reason \\
\midrule
1: Detection (\S\ref{sec:exp-detection}) & \texttt{ibm\_fez} QPU & hardware & 14 circuits, one job \\
2: Budget sweep (\S\ref{sec:exp-sample}) & Aer, \texttt{ibm\_fez} noise & simulated & 420 executions \\
3: Drift (\S\ref{sec:exp-drift}) & \texttt{ibm\_fez} QPU & hardware & 18 circuits, one job \\
\bottomrule
\end{tabular}
\end{table}

\section{Toolchain Integration and Deployment Outlook}
\label{app:discussion}
\subsection{Practical Relevance and the Maturation Trajectory}
\label{app:relevance}

A reasonable concern about any new verification framework is whether the threat it addresses is current or speculative. We address this directly. The three deployment scenarios in Section~\ref{sec:intro} (regulated industries, cloud QPU services, academic reproducibility) all rest on conditions that are present today rather than projected: production QML cloud services from IBM, IonQ, and Quantinuum already serve paying customers; the EU AI Act has been in force since August 2024 and NIST AI Risk Management Framework is in active use; the reproducibility crisis in classical ML is well documented \cite{hutson2018ai}. Each of these conditions transfers to QML as the technology moves from research prototypes to deployed systems.

The historical parallel is informative. Behavioural subtyping was formalised by Liskov and Wing \cite{liskov1994behavioral} in 1994 when object-oriented programming was still an emerging paradigm. The framework was not motivated by then-current attacks; it was motivated by the structural requirement that the discipline would be needed as the paradigm scaled. Three decades later it remains the foundational reference for behavioural subtyping. We position QML-PipeGuard similarly: the discipline is for QML deployment at the scale that current cloud roadmaps suggest will arrive within several years, and the framework is engineered to be available when that scale is reached rather than retrofitted after.

\subsection{Integration with Existing QML Toolchains}
\label{sec:integration}

The framework is designed to plug into existing QML toolchains rather than replace them. Three integration design sketches are immediate; a turnkey plugin for any specific toolchain is an engineering follow-up that we do not claim in the present paper.

For \textbf{Qiskit Machine Learning}, the verifier wraps the existing kernel and classifier modules at the SamplerV2 boundary. A pipeline using \texttt{QuantumKernel} or \texttt{VQC} continues to construct circuits as it does today; the verifier intercepts the Sampler job at submission, attaches the observable contract (Section~\ref{sec:fw-fp}), and routes a fraction of the shot budget to the contract observables. The reference implementation in this paper follows exactly this pattern: \texttt{src/ibm\_runtime.py} interacts with SamplerV2 in the same way a stock Qiskit pipeline would, and the verifier in \texttt{src/verification.py} is invoked without modifying the underlying QML model code.

For \textbf{PennyLane}, the equivalent hook is the \texttt{QNode} construction. PennyLane already exposes observables as first-class objects through \texttt{qml.expval}, so the contract observables can be evaluated by the same device on which the QML model runs. The compositionality result of Theorem~\ref{thm:composition} matches PennyLane's layer-by-layer execution model directly, allowing per-layer tolerances to be assigned where intermediate measurement is available.

For \textbf{TensorFlow Quantum}, integration is at the \texttt{tfq.layers.SampledExpectation} interface. The contract observables map to the same Cirq Pauli-string format used by the rest of the pipeline; the only addition is the per-job audit-log emission, which is straightforward to implement as a custom Keras callback.

In all three cases, the toolchain is treated as a black box and the framework operates at the boundary where measurement outcomes are produced. No re-implementation of encodings, ans\"atze, or classifiers is required. The reference code targets the protocol layer (SamplerV2 batched jobs) rather than any specific high-level wrapper; toolchain-specific adapters are a natural follow-on engineering deliverable.

\subsection{Anchoring the Reference}
\label{app:anchor}

A trusted prior run is not circular but a baseline-and-monitor discipline. The anchor is fixed once under witnessed conditions, for instance at acceptance testing or under audit. Every later round is then compared against it, so a provider that alters the channel afterwards is caught. This is the pattern of measured-boot attestation, where a baseline measurement is taken as given and later deviations are what the mechanism reports.

The stage specification $\sigma_A$ (Section~\ref{sec:background}) serialises to a JSON descriptor whose hash $H_{\text{spec}}$ anchors the audit trail. Each round appends a hash-chained record of $(t, \hat{\mu}_{t,\cdot}, \Delta_t, \text{decision})$. Any later alteration of a committed record breaks the chain and is caught on replay. The trail is tamper-evident, not authenticated (Section~\ref{sec:limitations}). The per-round classical overhead is negligible next to queue and execution time; the quantum overhead is exactly the budget $N$ of Theorem~\ref{thm:budget}, reported concretely in Section~\ref{sec:experiment}.

\section{Scope of the Local Family and Extended Observable Tiers}
\label{app:scope}

\subsection{What the Local Family Sees}
\label{app:scope-local}

The local Pauli family $\mathcal{O}_A^{\text{local}} = \{X_i, Y_i, Z_i : i = 1, \dots, n\}$ determines both single-qubit output marginals exactly and nothing more: its expectations are in bijection with the pair of Bloch vectors of the output marginals. Theorem~\ref{thm:detection} is the quantitative form of this statement, and Proposition~\ref{prop:local-ic} is its exact qualitative boundary for local-unitary substitutions. The substitution class covered by this family includes:
\begin{itemize}
\item \emph{Single-qubit gate substitutions}, that is, insertion of an extra rotation on any qubit (Pauli $X$, $Y$, $Z$, phase $S$, $T$, or arbitrary single-qubit unitary). The hardware experiment of Section~\ref{sec:exp-detection} uses this class with the $S$-gate variant.
\item \emph{Decoherence drift}, namely changes in $T_1$, $T_2$ between recalibrations, manifested as per-qubit channel shifts.
\item \emph{Bit/phase-flip channels}, i.e., adversarial or natural Pauli error channels acting independently on each qubit.
\item \emph{Per-qubit pulse manipulations}, that is, insider attacks at the pulse-control layer that target individual qubits.
\end{itemize}
These cover the dominant threat profile in current cloud QPU deployments, where pulse-level access is restricted and noise channels are predominantly local; detectability is decided by the commutation test of Proposition~\ref{prop:local-ic} against the deployed reference inputs.

\subsection{What It Cannot See}
\label{app:scope-out}

Substitutions that move only correlations are invisible to any marginal-determined family. Take the canonical example used throughout the paper. A channel outputs the Bell state $\ket{\Phi^+}$ on some reference input, and $Z \otimes I$ substitutes for it, swapping $\ket{\Phi^+}$ and $\ket{\Phi^-}$. The joint state becomes orthogonal (diamond distance $2$, trace norm $2$ on that input), yet both single-qubit marginals remain $I/2$. All six local-Pauli deviations are exactly zero, so $\delta_{\mathrm{loc}} = 0$. Detecting this class requires correlation observables: the single string $X_1 X_2$ distinguishes $\ket{\Phi^\pm}$ with expectation gap $2$. This particular substitution preserves $\langle Z_1 Z_2 \rangle = +1$, since both Bell states lie in the even-parity sector, and therefore leaves an inversion-test kernel entry unchanged. The kernel-relevant witness is instead $X \otimes I$, which maps $\ket{\Phi^+}$ to $\ket{\Psi^+} = (\ket{01}+\ket{10})/\sqrt{2}$: both marginals stay at $I/2$, so $\delta_{\mathrm{loc}} = 0$ and Proposition~\ref{prop:local-ic} marks it blind, while $\langle Z_1 Z_2 \rangle$ flips from $+1$ to $-1$ and the kernel entry falls from $0.5$ to $0$ (Section~\ref{sec:limitations}). The two sub-classes therefore need different observables: $X_1X_2$ catches the first, $Z_1Z_2$ the second, which is why the kernel tier of Section~\ref{sec:limitations} adds precisely $Z_1Z_2$. The cost of adding such observables is the subject of the tier analysis below.

Three classes of substitution fall outside the detection scope of the local family:
\begin{itemize}
\item \emph{Entangling substitutions} (e.g., insertion of a CNOT or CZ gate before measurement). These can alter joint correlation observables while leaving local marginals unchanged.
\item \emph{Correlated-noise channels} (e.g., crosstalk-induced $Z_1 Z_2$-correlated noise on superconducting hardware). The marginals may shift slightly, but the dominant signal sits in the correlation observables, which the local family cannot resolve.
\item \emph{Subspace-preserving stealthy substitutions} (e.g., a parity-preserving unitary that permutes the eigenstates of an entangled subspace). These leave local Pauli expectations fixed by construction and are invisible to the local family.
\end{itemize}

\subsection{Extended Tiers and Their Frame Constants}
\label{app:scope-ext}

We consider three nested families on two qubits and report their frame constants and resulting budgets at the operational parameters of Section~\ref{sec:experiment} ($\delta = 0.5$, $\varepsilon_A = 0.15$, $\eta = 0.05$, $B = 1$).

\begin{table}[tb]
\caption{Observable-family tiers on two qubits. $C$ is the frame constant of the tier's detection bound; $\gamma = \delta/C - \varepsilon_A$; $N$ is the sampled-reference budget at the deployed parameters. Tier-1 $C$ is exact and proved (Theorem~\ref{thm:detection}). The other constants are numerical values for the span-restricted witness objective, obtained by random sign-pattern search with SLSQP refinement (script \texttt{analysis/frame\_bound.py}). Since the search maximises, they are lower bounds on their tier's constant; every tier also obeys the proved bound $C \leq \sqrt{k}$. See the qualification below.}
\label{tab:tiers}
\centering
\renewcommand{\arraystretch}{1.2}
\small
\begin{tabular}{lcccc}
\toprule
\textbf{Tier} & $k$ & $C$ & $\gamma$ & $N$ \\
\midrule
Tier 1: local $\{X_i, Y_i, Z_i\}$ & 6 & $\sqrt{3} \approx 1.73$ (exact, tight) & $0.139$ & $13{,}680$ \\
Kernel tier: local $+$ $\{Z_1Z_2\}$ & 7 & $\approx 1.83$ (numerical) & $0.124$ & $\approx 20{,}700$ \\
Tier 2: local $+$ $\{X_1X_2, Y_1Y_2, Z_1Z_2\}$ & 9 & $\approx 2.21$ (numerical) & $0.076$ & $\approx 73{,}000$ \\
Tier 3: full two-qubit Pauli family & 15 & $\approx 3.73$ (numerical) & $< 0$ & --- \\
\bottomrule
\end{tabular}
\end{table}

Three qualifications govern Table~\ref{tab:tiers}. First, the Tier-1 constant is the exact, tight constant of Theorem~\ref{thm:detection}, valid for all $n$, and supports the marginal-level claim of that theorem. Second, the kernel-tier and Tier-2 constants normalise deviations lying in the \emph{span} of their own tier observables: the quantity they bound is $\sup\{\operatorname{Tr}(M\Delta) : M \in \operatorname{span}(\mathcal{O}),\, \|M\|_\infty \leq 1\}$, not the marginal separation of Definition~\ref{def:marginal-sep}. Directions outside that span, cross-correlations such as $X_1 Y_2$ or $Y_1 Z_2$, stay invisible. Both tiers therefore enlarge the detectable class, Tier 2 seeing the $\Phi^\pm$ swap through $X_1 X_2$, yet each keeps its own smaller blind set. Their $\gamma$ and $N$ entries reuse the Tier-1 target $\delta = 0.5$ so the rows stay comparable, and are indicative rather than certified: a constant that is under-stated over-states $\gamma$ and under-states $N$, so a deployment adopting a higher tier should recompute its own constant, or take the certified $C = \sqrt{k}$. Only the Tier-1 row rests on a proved constant. Bounding the joint trace-norm deviation over the same spans needs larger constants, $2.1701$ and $2.7321$. Third, the Tier-3 family spans the full traceless algebra on two qubits. Its detectable class (at tolerance $0$) is therefore every channel pair differing on some input of the deployed reference set $\mathcal{S}$. The price appears in the constant: the numerical estimate $C \approx 3.7320$ agrees numerically with $2+\sqrt{3}$ (conjectured, not proved) and is consistent with the upper bound $\sqrt{k} = \sqrt{15} \approx 3.87$ for a $k$-element orthogonal Pauli frame, and at the deployed parameters it drives $\gamma = 0.5/3.73 - 0.15 < 0$: the target separation is no longer detectable at this tolerance, and the deployment would have to either relax $\varepsilon_A$ below the drift floor (empty calibration interval, Section~\ref{sec:fw-driftcor}) or increase the target $\delta$. This is the concrete form of the completeness-versus-budget trade-off: richer families see more directions but certify each direction more weakly per shot, and the tier choice is a deployment decision.

\paragraph{Tier selection guidance.}
Tier 1 is the deployment default: exact constant, smallest budget, marginal-level guarantee, blind set characterised exactly and shrinkable through the reference-input set (Proposition~\ref{prop:local-ic}). The kernel tier is the natural default for kernel pipelines, whose output quantity contains $\langle Z_1 Z_2 \rangle$ (Section~\ref{sec:limitations}), at roughly $1.5\times$ the Tier-1 budget; its constant $\approx 1.8284$ agrees numerically with $2\sqrt{2}-1$ (conjectured, not proved). Tier 2 is indicated when the threat model includes correlation-only substitutions on further axis-aligned strings (e.g., Bell-basis swaps), at roughly $5.3\times$ the Tier-1 budget. Tier 3, at current tolerances and drift floors, is a diagnostic instrument rather than a runtime contract; it becomes deployable only if hardware drift or target separations improve substantially.

\begin{itemize}
\item \emph{Routine hardware-health monitoring} (e.g., a daily QPU-integrity check by a customer): Tier 1 suffices. Drift and single-qubit issues are the dominant failure modes, and the cost saving over higher tiers is meaningful for high-frequency monitoring.
\item \emph{Mid-stakes deployment} (e.g., healthcare or financial applications where the QML pipeline contains entangling gates and the threat model includes adversaries with limited pulse access): Tier 2 is the natural choice, adding the three diagonal-correlation Paulis at roughly $5.3\times$ the Tier-1 shot budget at the operational $\varepsilon_A = 0.15$. Deployments that wish to keep the shot budget closer to Tier 1 can raise the target separation $\delta$ so that $\gamma$ remains comparable to the Tier-1 value.
\item \emph{High-security deployment} (e.g., regulated infrastructure with full insider-threat consideration, or platforms exposing pulse-level access): Tier 3 spans every direction discriminable by Pauli measurements, but requires re-tuning the operational tolerance, since at $\varepsilon_A = 0.15$ and $\delta = 0.5$ the margin $\gamma$ is negative. Beyond this, full process tomography (with $O(4^{2n})$ cost) is the next level and is typically used for one-off certification rather than runtime monitoring.
\end{itemize}
The framework's structural guarantees (Theorem~\ref{thm:detection}, Proposition~\ref{prop:local-ic}, Theorem~\ref{thm:budget}, Corollary~\ref{cor:drift}) carry over unchanged across tiers; only the constant $C(\mathcal{O}_A)$, the family size $k$, and (at Tier 3) the operational tolerance are updated. This is intentional: the choice of family is a deployment-time configuration, not a framework redesign, and a verifier can move between tiers as the threat model evolves.

\paragraph{Cost scaling and allocation.}
Theorem~\ref{thm:budget} allocates $n_O = N/k$ shots to each observable, so its total is $\Theta(k \log k)$, with $k = 3n$ for the local family and $k = 4^n-1$ for the full Pauli family. That allocation is conservative. The local family splits into three commuting groups (all $X$, all $Y$, all $Z$), and one shot in a group estimates every observable in it, so $3n_O$ shots suffice at any $n$ and the total falls to $\Theta(\log k)$; the union bound is unchanged. The kernel and Tier-2 families also need three groups, the full two-qubit family five. Circuits per round follow the group count rather than $k$, and the runs reported here measured one observable per circuit, so they paid the factor. Classical cost is $k$ comparisons and one hash append per round. Classical shadows \cite{huang2020predicting} reach the same $\log k$ dependence with a weight factor of $3$ for single-qubit Paulis, at the price of randomised bases in place of a declared family, which would restate both the contract and the audit record. Adaptive allocation concentrates shots where the deviation sits, a minority here ($3$ of $6$ observables), at the price of reworking the union bound behind Theorem~\ref{thm:budget}.

\paragraph{Scaling in qubit count.}
The results hold for $n$ qubits, and $\sqrt{3}$ does not depend on $n$. Only the evaluation and the tier constants are two-qubit. The cost remains small: grouped, a round is about $11{,}720$ shots in three circuits at $n = 100$. Coverage shrinks: the family resolves $3n$ of $4^n-1$ directions. What does not change with $n$ is the exact resolution of every single-qubit marginal, and with it the reach over the classes of Appendix~\ref{app:scope-local}. Reference generation does not get harder either. A single-qubit expectation depends only on its backward light cone, so at constant depth and bounded entanglement degree it is classically computable at any $n$. The exception is a feature map that entangles every qubit, as \texttt{ZZFeatureMap} does. The open risk is crosstalk. It is correlated and tends to grow with $n$, so it may raise $d_{\mathrm{drift}}$ and close the interval $(0.067, 0.2887)$. We have no data there.

\bigskip
\begin{credits}
\subsubsection{\discintname} The author has no competing interests to declare.
\end{credits}

\bibliographystyle{splncs04}
\bibliography{qml_pipeguard}

\end{document}